%% file: main.tex
\documentclass{LMCS}
\input{packages.tex}

\input{macros.tex}

\def\doi{6 (3:20) 2010}
\lmcsheading%
{\doi}
{1--37}
{}
{}
{Oct.~30, 2009}
{Sep.~\phantom07, 2010}
{}   

\begin{document}

\title{Local Termination: theory and practice}

\author[J.~Endrullis]{J\"{o}rg Endrullis\rsuper a}
\address{{\lsuper{a,b}}%
    Vrije Universiteit Amsterdam,
    De Boelelaan 1081a,
    1081 HV Amsterdam,
    The Netherlands
}
\email{joerg@few.vu.nl,rdv@cs.vu.nl}

\author[R.~de Vrijer]{Roel de Vrijer\rsuper b}
\address{\vskip-6 pt}

\author[J.~Waldmann]{Johannes Waldmann\rsuper c}
\address{{\lsuper c}%
    Hochschule f\"ur Technik, Wirtschaft und Kultur Leipzig,
    Fakult\"at IMN, PF 30 11 66, D-04251 Leipzig, Germany
}
\email{waldmann@imn.htwk-leipzig.de}

\input{abstract}

\maketitle

\section{Introduction}

\input{intro}

\section{Preliminaries}

\input{prelims}

\section{Local Termination}\label{sec:sn}
\input{local}

\section{Local Relative Termination}\label{sec:snrel}
\input{localrel}

\section{Stepwise Removal of Rules}\label{sec:stepwise}
\input{stepwise}

\section{Via Models from Local to Global Termination}\label{sec:model}
\input{models}
\section{Starling and Owl}\label{sec:s}
\input{cl-s}

\section{The RFC Method}\label{sec:rfc}
\input{rfc}

\section{Monotone-models for Local Termination}\label{sec:qmodel}
\input{qmodels}

\section{Conclusion and Future Work}

\input{conclusion}

\bibliographystyle{alpha}
\bibliography{main}

\end{document}

%% file: packages.tex
\usepackage{graphicx}
\usepackage{latexsym}
\usepackage{amsmath,amssymb,amstext}
\usepackage{comment}
\usepackage{pifont}

\usepackage{enumerate}
\usepackage{color}
\usepackage{wrapfig}

\usepackage{indent}
\usepackage{hyperref}

%% file: macros.tex
\newcommand{\weg}[1]{}

\newcommand{\ignore}[1]{}

\theoremstyle{plain}
\newtheorem{theorem}[thm]{Theorem} %
\newtheorem{corollary}[thm]{Corollary}
\newtheorem{lemma}[thm]{Lemma}
\newtheorem{proposition}[thm]{Proposition}

\theoremstyle{definition}
\newtheorem{definition}[thm]{Definition}
\newtheorem{remark}[thm]{Remark}
\newtheorem{example}[thm]{Example}

\newtheorem{algorithm}[thm]{Algorithm}

\renewcommand{\paragraph}[1]{\vspace{2ex}\noindent\textbf{#1.} \hspace{.2em}}

\newcommand{\uRl}[1]{\texttt{{#1}}}

\makeatletter
\def\moverlay{\mathpalette\mov@rlay}
\def\mov@rlay#1#2{\leavevmode\vtop{%
   \baselineskip\z@skip \lineskiplimit-\maxdimen
   \ialign{\hfil$#1##$\hfil\cr#2\crcr}}}
\makeatother

\newcommand{\appsection}[1]{
  \let\oldthesection\thesection
  \renewcommand{\thesection}{Appendix \oldthesection}
  \section{#1}
  \let\thesection\oldthesection
}

\newcommand{\llbracket}{{[\![}}
\newcommand{\rrbracket}{{]\!]}}

\newcommand{\mit}{\mathit}

\newcommand{\mcl}{\mathcal}
\newcommand{\msf}{\mathsf}

\newcommand{\SN}{\msf{SN}}

\newcommand{\SNr}[1]{\SN_{#1}}
\newcommand{\SNrs}[1]{\funap{\SNr{#1}}}

\newcommand{\rel}[2]{\mathrel{#1/#2}}
\newcommand{\SNrel}[2]{\SN_{#1/#2}}
\newcommand{\SNrels}[2]{\funap{\SNrel{#1}{#2}}}

\newcommand{\sred}{{\rightarrow}}
\newcommand{\red}{\mathrel{\sred}}

\newcommand{\smred}{{\twoheadrightarrow}}
\newcommand{\mred}{\mathrel{\smred}}

\newcommand{\pairlft}{{\langle}}                %
\newcommand{\pairrgt}{{\rangle}}                %
\newcommand{\pairsep}{{,\,}}                    %
\newcommand{\pairstr}[1]{\pairlft#1\pairrgt}    %
\newcommand{\pair}[2]{\pairstr{#1\pairsep#2}}   %
\newcommand{\triple}[2]{\pair{#1\pairsep#2}}    %
\newcommand{\quadruple}[2]{\triple{#1\pairsep#2}} 

\newcommand{\tuple}{\pairstr}

\newcommand{\sjoin}{{\cup}}
\newcommand{\join}{\mathbin{\sjoin}}
\newcommand{\smeet}{{\cap}}
\newcommand{\meet}{\mathbin{\smeet}}
\newcommand{\sfunin}{{:}}
\newcommand{\funin}{\mathrel{\sfunin}}
\newcommand{\powerset}[1]{{2}^{#1}}

\newcommand{\setemp}{{\varnothing}}%
\newcommand{\srelcomp}{\cdot}
\newcommand{\relcomp}{\binap{\srelcomp}} %

\newcommand{\setsize}{\lstlength}

\newcommand{\mybind}[3]{#1#2.\:#3}

\newcommand{\myall}{\mybind{\forall}}

\newcommand{\mylam}[2]{\lambda#1.#2}

\newcommand{\nat}{\mathbb N}

\newcommand{\asig}{\Sigma}
\newcommand{\avars}{\mathcal{X}}
\newcommand{\ster}{\msf{Ter}}
\newcommand{\ter}{\funap{\ster}}

\newcommand{\svar}{\mit{Var}}

\newcommand{\vars}{\funap{\svar}}

\newcommand{\atrs}{R}
\newcommand{\btrs}{S}
\newcommand{\ctrs}{U}

\newcommand{\termset}{T}

\newcommand{\apos}{p}

\newcommand{\pos}{\funap{\mathcal{P}\!os}}

\newcommand{\asubst}{\sigma}
\newcommand{\subst}[2]{#2#1}

\newcommand{\ahole}{{[\hspace*{0.75pt}]}}%
\newcommand{\acontext}{C}
\newcommand{\contexthole}{\ahole}
\newcommand{\contextfill}[2]{#1[#2]}

\newcommand{\aalg}{\mcl{A}}
\newcommand{\balg}{\mcl{B}}
\newcommand{\calg}{\mcl{C}}
\newcommand{\aset}{A}
\newcommand{\bset}{B}
\newcommand{\cset}{C}
\newcommand{\arel}{\mathrel{R}}

\newcommand{\sinterpret}{\llbracket \cdot \rrbracket}
\newcommand{\interpret}[1]{\llbracket #1 \rrbracket}
\newcommand{\interpreta}[2]{\llbracket #1,#2 \rrbracket}
\newcommand{\alggt}{\succ}
\newcommand{\algge}{\sqsupseteq}
\newcommand{\algnu}{\leadsto}
\newcommand{\qge}{\ge}

\newcommand{\binap}[3]{#2\mathbin{#1}#3}
\newcommand{\funap}[2]{#1(#2)}

\newcommand{\bfunap}[3]{\funap{#1}{#2,#3}}

\newcommand{\where}{\mathrel{|}}

\newcommand{\sdefdby}{{=}}%
\newcommand{\defdby}{\mathrel{\sdefdby}}

\renewcommand{\implies}{\;\Rightarrow\;}

\newcommand{\en}{\;\;\&\;\;}

\newcommand{\sarity}{\sharp}%
\newcommand{\arity}{\funap{\sarity}}

\newcommand{\lstemp}{\bot}
\newcommand{\lstlength}[1]{|#1|}
\newcommand{\lstconcat}[2]{#1 #2}
\newcommand{\wordemp}{\varepsilon}

\newcommand{\punc}[1]{\:\text{#1}}

\newcommand{\samumap}{\mu}
\newcommand{\amumap}{\funap{\samumap}}

\newcommand{\combS}{\msf{S}}
\newcommand{\combK}{\msf{K}}
\newcommand{\combOwl}{\delta} %
\newcommand{\scombAt}{\msf{@}}
\newcommand{\combAt}{\bfunap{\scombAt}}

\newcommand{\famil}[1]{\funap{\mcl{F}\!am_{#1}}}
\newcommand{\family}{\famil{}}
\newcommand{\lang}{\funap{\mcl{L}}}

\newcommand{\slAbel}[2]{#1^{#2}}
\newcommand{\sdolAbel}[1]{\mit{lab}_{\!#1}}
\newcommand{\dolAbel}[1]{\bfunap{\sdolAbel{#1}}}
\newcommand{\lAbelsig}[1]{\funap{\sdolAbel{#1}}}
\newcommand{\lAbeltrs}[1]{\funap{\sdolAbel{#1}}}
\newcommand{\core}[1]{#1_{c}}
\newcommand{\nf}[2]{#1_{\funap{\mit{nf}}{#2}}}
\newcommand{\nfamumap}[1]{\samumap^{\funap{\mit{nf}}{#1}}}

\newcommand{\pto}{\rightharpoonup}
\newcommand{\defd}[1]{\left . #1 \right \downarrow}
\newcommand{\undefd}[1]{\left .#1\right\uparrow}
\newcommand{\peq}{\simeq}
\newcommand{\undefelt}{\bot}

\newcommand{\prog}{\mit}

\newcommand{\factorial}[1]{#1!}

%% file: abstract.tex
\begin{abstract}

The characterisation of termination using well-founded monotone algebras
has been a milestone on the way to automated termination techniques, 
of which we have seen an extensive development over the past years. 
Both the semantic characterisation and most known termination methods 
are concerned with \emph{global} termination, %
uniformly of all the terms of a term rewriting system (TRS). 
In this paper we consider \emph{local} termination, %
of specific sets of terms within a given TRS.

The principal goal of this paper is generalising the semantic
characterisation of global termination to local termination.
This is made possible by admitting the well-founded monotone algebras to be partial.
We also extend our approach to local relative termination.

The interest in local termination naturally arises in program
verification, where one is probably interested only in sensible inputs,
or just wants to characterise the set of inputs for which a program 
terminates.
Local termination will be also be of interest when dealing with 
a specific class of terms within a TRS that is 
known to be non-terminating, such as combinatory logic (CL) or a TRS encoding recursive
program schemes or Turing machines.

We show how some of the
well-known techniques for proving global termination, 
such as stepwise removal of rewrite rules and
semantic labelling, can be adapted to the local case.
We also describe transformations reducing local 
to global termination problems.
The resulting techniques for proving local termination have
in some cases already  been automated.

One of our  applications concerns the
characterisation of the terminating $\combS$-terms in CL as regular language. 
Previously this language had already been found 
via a tedious analysis of the reduction behaviour of $\combS$-terms. These 
findings have now been vindicated by a fully automated and verified proof.

\end{abstract}

%% file: intro.tex
\noindent An important contribution to the development
of automated methods for proving termination
has turned out to be
the characterization of termination using well-founded monotone algebras.
Both the semantic characterization and most known termination methods 
are concerned with \emph{global} termination, uniformly of 
all the terms of a TRS. 
This is remarkable, as termination is prima facie a property of individual terms. 
More generally, one may consider the termination problem
for an arbitrary set of terms of a TRS. 
We call this the \emph{local} termination problem.

A typical area where termination techniques are applied is that
of program verification.
The termination problems naturally arising in program verification
are local termination problems:
the central interest is termination of a program when started on a valid input.
A simple example of a program that is not globally terminating 
is the factorial function:
\begin{align*}
  \funap{\prog{fac}}{0} &= 1\\
  \funap{\prog{fac}}{n} &= n \cdot \funap{\prog{fac}}{n-1}
\end{align*}
This function terminates for all integers $n \ge 0$.
However, when started on a negative number this function
is caught in an infinite recursion. (This program will be used as an illustration in
Examples~\ref{ex:fac} and \ref{ex:fac2}.)

In logic programming (e.g. Prolog), 
local termination has been a central field of research over the past years.
Local termination problems of Haskell programs 
have been considered in~\cite{panitz97} and~\cite{Haskell06}.
In~\cite{panitz97}, a tableau calculus is devised to show termination
of sets of terms of the form $f\, a_1\, \ldots \,a_n$ where the $a_i$'s are in normal form.
In~\cite{Haskell06}, a transformation from Haskell programs into dependency 
pair problems~\cite{AG00} is given, which then in turn
are solved using methods for global termination.

Surprisingly, for TRSs not much work is known about local termination.
We mention the method of match-bounded string rewriting \cite{Matchbounds04},
which can be used to prove local termination for sets of strings generated by 
a regular automaton.
Indeed, this method can be viewed as an instance of the semantic framework
we develop in this paper.

Local termination is of special interest when dealing with 
specific classes of terms within a TRS that is known to be non-terminating.
Examples of such TRSs are combinatory logic (CL)~\cite{cl} and
encodings of recursive program schemes or Turing machines.
The well-known halting problem for Turing machines is a local termination problem.
Clearly, this holds for the blank tape halting problem which just
asks for termination on the blank tape.
On the first glance the uniform halting problem -- asking for termination on all
inputs -- might seem to be global.
However, this is a local termination problem as well, since Turing machines
are started in a distinguished initial state and admit only one head to work on the tape.
In this paper we will use CL and the halting problem for Turing machines to illustrate 
some of our results 
(Examples~\ref{ex:S},~\ref{ex:Smodel},~\ref{owl},~\ref{starling},~\ref{ex:SK} and~\ref{ex:TM}).

\paragraph{Outline and Contribution}
In Section~\ref{sec:sn} we generalize the semantic characterization from 
global termination to local termination based on well-founded, monotone partial $\asig$-algebras.
This establishes a first, important step towards the development of 
automatable techniques for proving local termination.
In Section~\ref{sec:snrel} we extend this to relative termination,
obtaining a characterization using extended monotone \mbox{partial $\asig$-algebras}.

For global termination it is common practice to stepwise simplify 
the proof obligation by removing rules.
For local termination (the strictly decreasing) rules cannot simply be removed
as they influence the set of reachable terms.
We need to impose weak conditions on the `removed' rules, see Section~\ref{sec:stepwise}.

Having developed the general framework, in the remaining sections we look for
fruitful instances of partial monotone algebras,
suitable for automation.

In Section~\ref{sec:model} we consider the case that the family of the set of terms 
for which we want to prove local termination can be described by a partial model.
 A variant of semantic labeling~\cite{Z95} can then be used
to transform the local termination problem into a global termination problem,
and the available provers for global termination can be applied.

In Section~\ref{sec:s} 
we consider TRSs with the property that strong and weak normalization coincide.
In particular, this holds for orthogonal, non-erasing TRSs.
In case the language of normalizing terms happens to be regular,
we show how a tree automaton (partial model) can be found 
accepting exactly the  normalizing terms.
Then we label the TRS with the obtained partial model,
and employ the theory developed in Section~\ref{sec:model}
to transform the local termination problem for the set
of normalizing terms to 
global termination of the labeled TRS.
We automated the search for the tree automaton as well as the labeling.

We apply this method to two well-known combinators from CL:
$\combS$ and $\combOwl$ with the rewrite rules 
$\combS\,xyz \to xz(yz)$ and $\combOwl\,xy \to y(xy)$, respectively.
Determining the language $N$ of normalizing $\combS$-terms
has been open until the year 2000~\cite{DBLP:journals/iandc/Waldmann00}.
Using the method from Section~\ref{sec:s}
we can now automatically find the partial model for $N$,
and we obtain a labeled system whose global termination
coincides with local termination on $N$.
Global termination of this labeled system (containing 1800 labeled rules) 
has been proven by TTT2~(1.0)~\cite{TTT09} and the proof has been verified by CeTA~(1.05)~\cite{ceta}.

In Section~\ref{sec:rfc} we demonstrate that the local termination method proposed in Section~\ref{sec:model}
can also be applied for proving global termination.
To that end, we transform the global termination into local termination
for the set of right-hand sides of forward closures~\cite{dershowitz:81}.
Then we transform the obtained system back into a global termination problem
using the transformation from Section~\ref{sec:model}.
We show the applicability of this method by solving an example 
that remained unsolved in the last termination competition~\cite{term:comp:08}.
After the transformation, the system allows
for a simple termination proof using linear polynomial interpretations.

In Section~\ref{sec:qmodel} we combine
the partial variant of the quasi-models of~\cite{Z95} with 
monotone algebras to obtain partial monotone algebras.
Roughly speaking, partial quasi-models are deterministic tree automata~\cite{tata:07}
equipped with a relation $\ge$ on the states which guarantees 
that the language of the automaton is closed under rewriting.
Thereby we obtain partial monotone algebras
that can be applied successfully for proofs of local termination.
Indeed, this method can be automated and, as a matter of fact,
we have devised an implementation.

A preliminary version of this paper has appeared in~\cite{endr:vrij:wald:09};
our additional contribution is as follows:
\begin{enumerate}[$\bullet$]
  \item 
    For TRSs where strong and weak normalization coincide 
    and where the language of normalizing terms $N$ is regular, 
    we describe an algorithm for constructing
    a tree automaton (a partial model) accepting exactly the language $N$.
    For example, this method is applicable for (fully automatically) determining the
    language of normalizing $\combS$-terms.
  \item
    We show that methods for local termination can fruitfully be employed
    for proving global termination, classically the main focus of termination analysis for TRSs.
    Employing the RFC method (right-hand sides of forward closures)
    in combination with the transformation from local to global termination
    from Section~\ref{sec:model},
    we give a new proof for a string rewrite system (SRS) for which
    no proof had been found in the termination competition so far.
\end{enumerate}

%% file: prelims.tex
\subsection*{Term rewriting}

\noindent A \emph{signature} $\asig$ is a non-empty set of symbols,
 each having a fixed \emph{arity},
given by a map $\sarity \funin \asig \to \nat$. 
Let $\asig$ be a signature and $\avars$ a set of variable symbols.
The \emph{set $\ter{\asig,\avars}$ of terms over $\asig$ and $\avars$} 
is the smallest set satisfying:
$\avars \subseteq \ter{\asig,\avars}$, and
$f(t_1,\dots,t_n) \in \ter{\asig,\avars}$ if $f \in \asig$ with arity $n$ 
and $\forall i (1 \leq i \leq n): t_i \in \ter{\asig,\avars}$.
We use $x,y,z,\ldots$ to range over variables.
The set of positions $\pos{t} \subseteq \nat^*$ of a term $t \in \ter{\asig,\avars}$
is defined as follows:
$\pos{\funap{f}{t_1,\ldots,t_n}} = 
 \{\lstemp\} \cup \{\lstconcat{i}{\apos} \where 1 \le i \le \arity{f},\, \apos \in \pos{t_i}\}$
and
$\pos{x} = \{\lstemp\}$ for variables $x \in \avars$.

A substitution $\asubst$ is a map $\asubst : \avars \to \ter{\asig,\avars}$.
For a term $t \in \ter{\asig,\avars}$ %
we define $\subst{\asubst}{t}$ as the result of replacing each $x \in \avars$
in $t$ by $\funap{\asubst}{x}$.
Formally, $\subst{\asubst}{t}$ is 
inductively defined by
$\subst{\asubst}{x} \defdby \funap{\asubst}{x}$ for variables $x \in \avars$
and otherwise
$\subst{\asubst}{\funap{f}{t_1,\ldots,t_n}} \defdby \funap{f}{\subst{\asubst}{t_1},\ldots,\subst{\asubst}{t_n}}$.
Let $\contexthole$ be a fresh %
symbol, $\contexthole \not\in \asig\join\avars$.
A \emph{context} $\acontext$ is a term from $\ter{\asig,\avars\join\{\contexthole\}}$
containing precisely one occurrence of $\contexthole$.
By $\contextfill{\acontext}{s}$ we denote the term $\subst{\asubst}{\acontext}$
where $\funap{\asubst}{\contexthole} = s$ and $\funap{\asubst}{x} = x$ 
for all $x \in \avars$.

A \emph{term rewriting system (TRS)} $R$ over $\asig$ and $\avars$ is a set of
pairs $\pair{\ell}{r} \in \ter{\asig,\avars}$,
called \emph{rewrite rules} and written as $\ell \to r$,
for which 
the \emph{left-hand side} $\ell$ is not a variable ($\ell \not\in \avars$),
and all variables in the \emph{right-hand side} $r$ occur in $\ell$ as well
($\vars{r} \subseteq \vars{\ell}$). 
Let $\atrs$ be a TRS.
For terms $s, t \in \ter{\asig,\avars}$ we write $s \to_\atrs t$ 
 (or briefly $s \to t$)
if there exists a rule $\ell \to r \in \atrs$, a substitution $\asubst$
and a context $\acontext \in \ter{\asig,\avars\join\{\contexthole\}}$
such that
$s = \contextfill{\acontext}{\subst{\asubst}{\ell}}$
and
$t = \contextfill{\acontext}{\subst{\asubst}{r}}$. 
The reflexive--transitive closure of $\to$ is denoted by $\mred$.
We call $\to$  the \emph{one-step rewrite relation} induced by $\atrs$
and  $\mred$ the \emph{many-step rewrite} or \emph{reduction relation}.
If $t \mred t'$ then we call $t'$ an \emph{($R$)-reduct} of $t$.

\begin{definition}\normalfont\label{def:local:family}
  Let $\atrs$ be a TRS over $\asig$ and $\termset \subseteq \ter{\asig,\avars}$ a set of terms.
  The \emph{family $\famil{\atrs}{\termset}$ of $\termset$} 
  is the set of (not necessarily proper) subterms of $R$-reducts of terms $t \in \termset$
  (that is, the least set containing $t$ that is closed under reduction and taking subterms).
\end{definition}

\subsection*{Partial functions}
For partial functions $f \funin \aset_1 \times \ldots \times \aset_n \pto \aset$
and $a_1 \in \aset_1$, \ldots, $a_n \in \aset_n$
we call $\funap{f}{a_1,\ldots,a_n}$ \emph{defined} and write 
$\defd{\funap{f}{a_1,\ldots,a_n}}$ 
whenever $\tuple{a_1,\ldots,a_n}$ is in the domain of $f$.
Otherwise $\funap{f}{a_1,\ldots,a_n}$ is called \emph{undefined} and we write 
$\undefd{\funap{f}{a_1,\ldots,a_n}}$. We use the same terminology and notation 
for composite expressions involving partial functions. Between such expression 
we use Kleene equality: 
\begin{align*}
  \mathit{exp_1} \peq \mathit{exp_2} \;\;\Longleftrightarrow_{\mit{def.}}\;\;
  (\undefd{\mathit{exp_1}} \text{ and } \undefd{\mathit{exp_2}}) \text{ or }
  (\defd{\mathit{exp_1}} \text{ and } \defd{\mathit{exp_2}} \text{ and } \mathit{exp_1} = \mathit{exp_2})
\end{align*}
Note that an expression can only be defined if all its subexpressions are. 

\begin{definition}\normalfont\label{def:monotone}
Let ${\aset}$ be a set and $\arel$ a relation on ${\aset}$.
We define two properties of an $n$-ary partial function $f$ with respect to $\arel$.
\begin{enumerate}[(1)]
\item 
$f$ is \emph{closed}
if for every $a, b \in \aset$ we have:
\[
\defd{\funap{f}{\ldots,a,\ldots}}  \en a \arel b
    \implies 
\defd{\funap{f}{\ldots,b,\ldots}} 
\]
\item
$f$ is \emph{monotone}
if for every $a, b \in \aset$ we have:
\[
\defd{\funap{f}{\ldots,a,\ldots}}  \en \defd{\funap{f}{\ldots,b,\ldots}}  \en a \arel b
    \implies 
\funap{f}{\ldots,a,\ldots} \arel \funap{f}{\ldots,b,\ldots}
\]
\end{enumerate}
\end{definition}

The functions that we  consider will be typically both
closed and monotone, which can be rendered  briefly as:
\[
\defd{\funap{f}{\ldots,a,\ldots}}  \en a \arel b
    \implies 
\funap{f}{\ldots,a,\ldots} \arel \funap{f}{\ldots,b,\ldots}
\]
By writing something like $\mathit{exp_1} \arel \mathit{exp_2}$
we   imply that $\mathit{exp_1}$ and $\mathit{exp_2}$ are defined.

\subsection*{Partial \texorpdfstring{$\boldsymbol{\asig}$}{}-algebras}
We give the definition of a partial algebra:
\begin{definition}\normalfont
  A \emph{partial $\asig$-algebra} $\pair{\aset}{\sinterpret}$ consists of 
  a non-empty set $\aset$
  and for each $n$-ary $f \in \asig$ a partial function $\interpret{f} \funin A^n \pto A$, 
  the \emph{interpretation of $f$}.
  A \emph{$\asig$-algebra} $\pair{\aset}{\sinterpret}$ is a partial $\asig$-algebra $\pair{\aset}{\sinterpret}$
  where all interpretations $\interpret{f}$ for $f \in \asig$ are total.

  Given a partial $\asig$-algebra $\aalg = \pair{\aset}{\sinterpret}$
  and a (partial) assignment of the variables, $\alpha \funin \avars \pto \aset$,
  we can  give an \emph{interpretation} $\interpreta{t}{\alpha}$ 
  of terms $t \in \ter{\asig,\avars}$,
  which,  however, will not always be defined.
  So the interpretation is a partial function from terms and partial assignments to $A$,
  inductively defined by:
  \begin{align*}
    \interpreta{x}{\alpha} &\defdby \funap{\alpha}{x} \\
    \interpreta{\funap{f}{t_1,\ldots,t_n}}{\alpha} &\defdby
         \funap{\interpret{f}}{\interpreta{t_1}{\alpha},\ldots,\interpreta{t_n}{\alpha}}
  \end{align*}
  For ground terms $t \in \ter{\asig,\setemp}$ we write $\interpret{t}$ for short.
\end{definition}
Whenever a term $t$ is defined, that is, $\defd{\interpret{t}}$,
then all subterms of $t$ are defined as well.
This is a consequence of the usual definition of the composition of partial functions (functional relations);
`undefined' is not an element of the domain.
We say that a set of terms $\termset$ is defined if all terms in $\termset$ are defined:

\begin{definition}\normalfont
A set $\termset \subseteq \ter{\asig,\setemp}$ is called \emph{defined} if 
   for all $t \in   \termset$ we have $\defd{\interpret{t}}$.
\end{definition}

\subsection*{Partial models}
First, we generalise the models from~\cite{Z95} to partial models.

\begin{definition}
  A \emph{model} for a TRS $\atrs$ is a $\asig$-algebra $\pair{\aset}{\sinterpret}$
  such that
  $\interpreta{\ell}{\alpha} = \interpreta{r}{\alpha}$ 
  for every rule $\ell \to r \in \atrs$ and every interpretation $\alpha \funin \vars{\ell} \to \aset$ of the variables.
\end{definition}
\begin{definition}\normalfont\label{def:local:model}
  Let $\atrs$ be a TRS over $\asig$. 
  A \emph{partial model $\aalg \defdby \pair{\aset}{\sinterpret}$ for $\atrs$} 
  is a partial $\asig$-algebra $\aalg$,
  such that
  \begin{align*}
  \defd{\interpreta{\ell}{\alpha}} \implies \interpreta{\ell}{\alpha} = \interpreta{r}{\alpha}
  \end{align*}
  for every $\ell \to r \in \atrs$ and $\alpha \funin \vars{\ell} \to \aset$.
\end{definition}
Thus the condition $\interpreta{\ell}{\alpha} = \interpreta{r}{\alpha}$ of models
is only required if
the interpretation of the left-hand side is defined, that is, $\defd{\interpreta{\ell}{\alpha}}$.
Observe that a left-hand side $\ell$ may be undefined while the corresponding right-hand side $r$ is defined;
the other way around is not permitted.
This asymmetry is crucial, 
since rewriting may turn an undefined (non-terminating) term into a defined (terminating) term,
but not the other way around.

\begin{definition}\normalfont\label{def:modellang}
  Let $\aalg = \pair{\aset}{\sinterpret}$ be a partial model.
  The \emph{language $\lang{\aalg}$ of $\aalg$} is:
  \begin{align*}
  \lang{\aalg} = \{t \in \ter{\asig,\setemp} \where \defd{\interpret{t}}\}
  \end{align*}
\end{definition}

We further generalise the concept of partial models to relations:
\begin{definition}\normalfont\label{def:dec}
  Let $\atrs$ be a TRS over $\asig$, $\aalg = \pair{\aset}{\interpret{\cdot}}$ be a partial $\asig$-algebra,
  and ${\alggt} \subseteq \aset \times \aset$ a binary relation.
  We say that $\triple{\aset}{\interpret{\cdot}}{{\alggt}}$ is a \emph{partial model for $\atrs$} if:
  \begin{align*}
   \defd{\interpreta{\ell}{\alpha}} 
   \implies
   \interpreta{\ell}{\alpha} \alggt \interpreta{r}{\alpha}
  \end{align*}
  for all ${\ell \to r} \in \atrs$ and every assignment $\alpha \funin \vars{\ell} \to \aset$.
  If additionally $\pair{\aset}{\interpret{\cdot}}$ is a (total) $\asig$-algebra, then
  $\triple{\aset}{\interpret{\cdot}}{{\alggt}}$ is called a \emph{model for $\atrs$}.
  Whenever the partial $\asig$-algebra $\aalg$ is clear from the context,
  we say `\emph{$\alggt$ is a partial model for $\atrs$}' for short.
\end{definition}

Note that $\pair{\aset}{\interpret{\cdot}}$ is a partial model
if and only if $\triple{\aset}{\interpret{\cdot}}{{=}}$ is a partial model.

\begin{remark}
The notion of partial model $\triple{\aset}{\interpret{\cdot}}{{\alggt}}$
is closely related to that of quasi-model~\cite{Z95}.
In particular,
a quasi-model for a TRS $R$ is a (total)
monotone model $\triple{\aset}{\interpret{\cdot}}{{\ge}}$ for $R$.
\end{remark}

\begin{remark}
  In Definition~\ref{def:local:model} and~\ref{def:dec} we could as well quantify over partial assignments $\alpha \funin \avars \pto \aset$
  in place of total assignment $\alpha \funin \avars \to \aset$.
  This gives rise to an equivalent definition as an undefined value for a variable in the left-hand side $\ell$ (and $\vars{r} \subseteq \vars{\ell}$)
  would result in $\interpreta{\ell}{\alpha}$ being undefined, and
  thereby invalidate the precondition $\defd{\interpreta{\ell}{\alpha}}$ of the implication.
\end{remark}

%% file: local.tex
\noindent We devise a complete characterization of local termination
based on an extension of the monotone algebra approach of~\cite{EWZ06,Z94}.
The central idea is the use of monotone partial algebras,
that is, the operations of the algebras are allowed to be partial functions.
This idea was introduced in
 \cite{endr:grab:hend:isih:vrij:08}, where these algebras have been employed
to obtain a complete characterization of local infinitary strong normalization.
First we give the definition of local termination:

\begin{definition}\normalfont
  Let $A$ be a set and ${\to} \subseteq {A \times A}$ a binary relation on $A$.
  Then $\to$ is called \emph{terminating on} $B \subseteq A$ if
  no $b \in B$ admits an infinite sequence 
  \begin{align*}
    b = b_1 \to b_2 \to \ldots 
  \end{align*}
  Note that $b = b_1 \in B$. The elements $b_2$, $b_3$, \ldots, however, may or may not be in $B$.
\end{definition}

\begin{definition}\normalfont
  A TRS $\atrs$ over $\asig$
  is called \emph{terminating (or strongly normalizing) on $\termset \subseteq \ter{\asig,\avars}$}, 
  denoted $\SNrs{\atrs}{\termset}$,
  if $\red_{\atrs}$ is terminating on $\termset$.
  We write $\SNr{\atrs}$ for termination on the set of all terms $\ter{\asig,\avars}$.
\end{definition}

We introduce the concept of monotone partial $\asig$-algebras.
In contrast with~\cite{endr:vrij:wald:09} we do not require well-foundedness of $\alggt$.
We think that it is conceptually cleaner 
to distinguish the two concepts.
Monotone partial $\asig$-algebras $\triple{\aset}{\interpret{\cdot}}{{\alggt}}$ for
which well-foundedness of $\alggt$ holds, will be called well-founded.

\begin{definition}\normalfont\label{def:local:alg}
  A \emph{monotone partial $\asig$-algebra $\triple{\aset}{\interpret{\cdot}}{{\alggt}}$}
  is a partial $\asig$-algebra $\pair{\aset}{\sinterpret}$ 
  equipped with a binary relation $\alggt \subseteq \aset \times \aset$ on $\aset$ such that
  for every $f \in \asig$ the function $\interpret{f}$ is closed and monotone with respect to $\alggt$.
  
  A monotone partial $\asig$-algebra $\triple{\aset}{\interpret{\cdot}}{{\alggt}}$ is called \emph{well-founded}
  if  $\alggt$ is well-founded.
\end{definition}

\begin{remark}
One could also work with monotone, \emph{total} algebras instead of partial algebras, 
by adding an ``undefined'' element $\undefelt$ to the domain.
Then defining $\undefelt$ to be maximal, $\undefelt \alggt a$ 
for every $a \in A \setminus \{\undefelt\}$, monotonicity of
a function will automatically entail closedness.
In order to get full correspondence
with our framework of partial algebras,
we would in this  set-up only consider strict functions
(that is, the value of the function is $\undefelt$ whenever one of the arguments is $\undefelt$).
\end{remark}

The following theorem gives a complete characterization of local termination
in terms of monotone partial algebras.
\begin{theorem}\label{thm:sn}
  Let $\atrs$ be a TRS over $\asig$, and $\termset \subseteq \ter{\asig,\setemp}$.
  Then $\SNrs{\atrs}{\termset}$ holds if and only if
  there exists a well-founded monotone partial $\asig$-algebra 
  $\aalg = \triple{\aset}{\sinterpret}{{\alggt}}$
  such that $\termset$ is defined,
  and $\alggt$ is a partial model for $\atrs$.
\end{theorem}
\begin{proof}
  Theorem~\ref{thm:sn} is proved in the same way as Theorem~\ref{thm:snrel} (using Remark~\ref{rem:relterm}).
\end{proof}
To keep the presentation simple, the theorem characterizes local termination for sets
of ground terms $\termset \subseteq \ter{\asig,\setemp}$ only.
Indeed, the theorem can easily be generalized to sets of open terms
by, instead of just a well-founded monotone partial algebra, 
additionally requiring a variable assignment $\alpha$.
A  set of terms $\termset$ is then called  defined if for that $\alpha$ we have
$\defd{\interpreta{t}{\alpha}}$ for every $t \in \termset$.

\begin{remark}
  In case $T = \ter{\asig,\setemp}$ is the set of all ground terms, Theorem~\ref{thm:sn}
  basically coincides with the usual theorem for proving termination using 
  (total) well-founded monotone $\asig$-algebras.
  More precisely, the subalgebra of $\aalg$ containing all elements that are interpretations
  of ground terms (leaving out the junk) is a (total) well-founded monotone $\asig$-algebra
  proving termination of $\atrs$.
\end{remark}

\begin{example}\label{ex:S}
  We consider the $\combS$ combinator with the rewrite rule 
  \begin{align*}
  \combS xyz \to xz(yz)
  \end{align*}
  from combinatory logic. That is:
  \begin{align*}
    \combAt{\combAt{\combAt{\combS}{x}}{y}}{z} \to \combAt{\combAt{x}{z}}{\combAt{y}{z}}
  \end{align*}
  in first order notation.
  The $\combAt{M}{N}$ is abbreviated by $MN$.

  The $\combS$ combinator is known to be globally non-terminating.
  For example the term $\combS (\combS \combS) (\combS \combS) (\combS (\combS \combS) (\combS \combS))$
  admits an infinite reduction, see further~\cite{zachos:78}.
  We have, however, local termination on certain sets of terms,
  for example the set 
  of ``flat'' $\combS$-terms:
  \begin{align*}
    \termset \defdby \{\combS^n \where n \in \nat,\, n \ge 1\}
  \end{align*}
  where $\combS^1 \defdby \combS$ and $\combS^{n+1} \defdby \combAt{\combS^n}{\combS}$.

\newcommand{\ints}{\msf{s}}
  We prove strong normalization on $\termset$ 
  using the well-founded monotone partial $\asig$-algebra $\aalg = \triple{\aset}{\sinterpret}{{\alggt}}$,
  where $\aset \defdby \{\ints\} \cup \nat$ and the interpretation $\sinterpret$ is given by:
  \begin{align*}
    \interpret{\combS} &\defdby \ints &
    \bfunap{\interpret{\scombAt}}{\ints}{\ints} &\defdby 0 &
    \bfunap{\interpret{\scombAt}}{0}{n} &\defdby n + 1&
    \bfunap{\interpret{\scombAt}}{n}{\ints} &\defdby 2\cdot n + 1
  \end{align*}
  for all $n \in \nat$ and  $\undefd{\bfunap{\interpret{\scombAt}}{x}{y}}$
  for all other cases.
  Let $\alggt$ be the natural order on $\nat$;
  that is, $\ints$ is neither source nor target of a $\alggt$ step.
  Then well-foundedness of $\alggt$ and monotonicity of $\interpret{\scombAt}$
  are obvious, and $\termset$ is defined.
  We have
  $\defd{\interpreta{\combS xyz}{\alpha}}$
  only if $\funap{\alpha}{x} = \ints$
  and $\funap{\alpha}{z} = \ints$;
  then we obtain:
  \begin{align*}
    \interpreta{\combS xyz}{\alpha} = 3 &\alggt 1 = \interpreta{xz(yz)}{\alpha} &&\text{for $\funap{\alpha}{y} = \ints$}\\
    \interpreta{\combS xyz}{\alpha} = 2\cdot \funap{\alpha}{y} + 3 &\alggt 2\cdot \funap{\alpha}{y} + 2 
    = \interpreta{xz(yz)}{\alpha} &&\text{for $\funap{\alpha}{y} \in \nat$}
  \end{align*}
  Hence $\alggt$ is a partial model for $\combS xyz \to xz(yz)$ and we conclude termination on $\termset$.
\end{example}

\newcommand{\fs}{\msf{s}}
\newcommand{\fm}{\msf{-}}
\newcommand{\mul}{\prog{mul}}
\newcommand{\add}{\prog{add}}
\newcommand{\inv}{\prog{inv}}
\begin{example}\label{ex:fac}
We recall the Haskell program from the introduction:
\begin{align*}
  \funap{\prog{fac}}{0} &\ {::}\ \msf{Integer} \to \msf{Integer}\\
  \funap{\prog{fac}}{0} &= 1\\
  \funap{\prog{fac}}{n} &= n \cdot \funap{\prog{fac}}{n-1}
\end{align*}
We remark that the standard Haskell data type $\msf{Integer}$ allows for negative numbers.
For this reason the program is not globally terminating, but only locally on non-negative integers.
The usual implementation of the factorial function as TRS
makes use of Peano numerals for encoding natural numbers using a constant `$0$' and a unary symbol `$\fs$' for successor.
Then the problem of negative numbers does not occur.

For the purpose of modeling the Haskell program as close as possible,
we have chosen for a different encoding of the factorial function as TRS.
For encoding negative numbers we extend Peano numerals with a unary symbol `$\fm$'.
Since standard term rewriting does not allow for a priority order on rules,
we need to dissolve ambiguities, that is, overlaps between the rules,
by instantiating the variables; e.g.\ for the factorial function $\prog{fac}$
the variable $n$ needs to be instantiated with $\funap{\fs}{n}$ and $\funap{\fm}{n}$
to match exactly the integers (in this case $0$) not covered by the first rule.
As the result of the translation we obtain the TRS $\atrs$:
\begin{align*}
  \funap{\prog{fac}}{0} &\to \funap{\fs}{0} \tag{$\rho_1$}\label{rule:fac:z}\\
  \funap{\prog{fac}}{\funap{\fs}{x}} &\to \bfunap{\mul}{\funap{\fs}{x}}{\funap{\prog{fac}}{x}} \tag{$\rho_2$}\label{rule:fac:s}\\
  \funap{\prog{fac}}{\funap{\fm}{x}} &\to \bfunap{\mul}{\funap{\fm}{x}}{\funap{\prog{fac}}{\funap{\fm}{\funap{\fs}{x}}}} \tag{$\rho_3$}\label{rule:fac:m}\\[1ex]
  \bfunap{\mul}{x}{0} &\to 0 \tag{$\rho_4$}\label{rule:mul:zr}\\
  \bfunap{\mul}{0}{y} &\to 0 \tag{$\rho_5$}\label{rule:mul:zl}\\
  \bfunap{\mul}{x}{\funap{\fs}{y}} &\to \bfunap{\add}{ \bfunap{\mul}{x}{y} }{x} \tag{$\rho_6$}\label{rule:mul:s}\\
  \bfunap{\mul}{\funap{\fs}{x}}{\funap{\fm}{y}} &\to \funap{\fm}{\bfunap{\mul}{\funap{\fs}{x}}{y}} \tag{$\rho_7$}\label{rule:mul:sm}\\
  \bfunap{\mul}{\funap{\fm}{x}}{\funap{\fm}{y}} &\to \bfunap{\mul}{x}{y} \tag{$\rho_9$}\label{rule:mul:mm}\\[1ex]
  \bfunap{\add}{x}{0} &\to x \tag{$\rho_{10}$}\label{rule:add:zr}\\
  \bfunap{\add}{0}{y} &\to y \tag{$\rho_{11}$}\label{rule:add:zl}\\
  \bfunap{\add}{x}{\funap{\fs}{y}} &\to \funap{\fs}{ \bfunap{\add}{x}{y} } \tag{$\rho_{12}$}\label{rule:add:s}\\
  \bfunap{\add}{\funap{\fs}{x}}{\funap{\fm}{\funap{\fs}{y}}} &\to \bfunap{\add}{x}{\funap{\fm}{y}} \tag{$\rho_{13}$}\label{rule:add:sm}\\
  \bfunap{\add}{\funap{\fs}{x}}{\funap{\fm}{0}} &\to \funap{\fs}{x} \tag{$\rho_{14}$}\label{rule:add:sm0}\\
  \bfunap{\add}{\funap{\fm}{x}}{\funap{\fm}{y}} &\to \funap{\fm}{\bfunap{\add}{x}{y}} \tag{$\rho_{15}$}\label{rule:add:mm}
\end{align*}
This TRS is globally non-terminating due to the rewrite sequence:
\begin{align*}
\funap{\prog{fac}}{\funap{\fm}{x}} 
 \to\ &\bfunap{\mul}{\funap{\fm}{x}}{\funap{\prog{fac}}{\funap{\fm}{\funap{\fs}{x}}}}\\
 \to\ &\bfunap{\mul}{\funap{\fm}{x}}{\bfunap{\mul}{\funap{\fm}{x}}{\funap{\prog{fac}}{\funap{\fm}{\funap{\fs}{\funap{\fs}{x}}}}}}
 \to\ \ldots
\end{align*}
We prove local termination on the set $T = \{\funap{\prog{fac}}{\funap{\fs^n}0} \where n \in \nat\}$.
Let $\aalg = \triple{\nat}{\sinterpret}{{>}}$ where $>$ is the natural order on $\nat$,
and the interpretation $\sinterpret$ is given by:
  \begin{align*}
    \interpret{0} &\defdby 0 &
    \funap{\interpret{\fs}}{n} &\defdby n+1 &
    \funap{\interpret{\prog{fac}}}{n} &\defdby \factorial{(2n+2)} \\
    \bfunap{\interpret{\mul}}{n}{m} &\defdby 2(n+1)(m+1) &
    \bfunap{\interpret{\add}}{n}{m} &\defdby n + 2m + 1&
    &\undefd{\bfunap{\interpret{\fm}}{n}{m}}
  \end{align*}
for all $n,m \in \nat$.
For all left-hand sides $\ell$ 
of $\eqref{rule:fac:m}$, $\eqref{rule:mul:sm}$, $\eqref{rule:mul:mm}$, $\eqref{rule:add:sm}$, $\eqref{rule:add:sm0}$, $\eqref{rule:add:mm}$
and all $\alpha \funin \avars \to \nat$ we have $\undefd{\interpreta{\ell}{\alpha}}$; thus $>$ is a partial model for these rules.
For the remaining rules we have:
\begin{align*}
  \interpreta{\funap{\prog{fac}}{0}}{\alpha} = 2 
    &> 1 = \interpreta{\funap{\fs}{0}}{\alpha}\\
  \interpreta{\funap{\prog{fac}}{\funap{\fs}{x}}}{\alpha} = \factorial{(2\funap{\alpha}{x}+4)} 
    = (2\funap{\alpha}{x}+4) &\cdot \factorial{(2\funap{\alpha}{x}+3)}
    >\\ 2 (\factorial{(2\funap{\alpha}{x}+2)}+1) (\funap{\alpha}{x}+2) &= \interpreta{\bfunap{\mul}{\funap{\prog{fac}}{x}}{\funap{\fs}{x}}}{\alpha}\\
  \interpreta{\bfunap{\mul}{x}{0}}{\alpha} = 2(\funap{\alpha}{x}+1)
    &> 0 = \interpreta{0}{\alpha} \\
  \interpreta{\bfunap{\mul}{0}{y}}{\alpha} = 2(\funap{\alpha}{y}+1)
    &> 0 = \interpreta{0}{\alpha} \\
  \interpreta{\bfunap{\mul}{x}{\funap{\fs}{y}}}{\alpha} = 2(\funap{\alpha}{x}+1)(\funap{\alpha}{y}+2)
    &>\\ 2(\funap{\alpha}{x}+1)(\funap{\alpha}{y}+1) + \funap{\alpha}{x} + 1 &= \interpreta{\bfunap{\add}{ \bfunap{\mul}{x}{y} }{x}}{\alpha}\\
  \interpreta{\bfunap{\add}{x}{0}}{\alpha} = \funap{\alpha}{x}+1 
    &> \funap{\alpha}{x} = \interpreta{x}{\alpha} \\
  \interpreta{\bfunap{\add}{0}{y}}{\alpha} = 2\funap{\alpha}{y}+1 
    &> \funap{\alpha}{y} = \interpreta{y}{\alpha} \\
  \interpreta{\bfunap{\add}{x}{\funap{\fs}{y}}}{\alpha} = \funap{\alpha}{x} + 2(\funap{\alpha}{y}+1) + 1 
    &> \funap{\alpha}{x} + 2\funap{\alpha}{y} + 2 = \interpreta{\funap{\fs}{ \bfunap{\add}{x}{y} } }{\alpha}
\end{align*}
for all $\alpha \funin \avars \to \nat$. Hence $>$ is a partial model for all rules in $\atrs$.
Moreover, $T$ is defined (that is, $\defd{T}$) since $\interpret{\funap{\prog{fac}}{\funap{\fs^n}{0}}} = \factorial{(2n + 2)} \in \nat$.
By Theorem~\ref{thm:sn} we conclude $\SNrs{\atrs}{T}$, that is, $R$ is terminating on $T$.
\end{example}

%% file: localrel.tex
We define local relative termination.
\begin{definition}\normalfont
  Let $A$ be a set and ${\to_1},\,{\to_2} \subseteq {A \times A}$ binary relations.
  Then $\to_1$ is called \emph{terminating relative to $\to_2$ on} $B \subseteq A$,
  denoted $\SNrels{\to_1}{\to_2}{B}$,
  if ${\rel{\to_1}{\to_2}} \defdby {\relcomp{\mred_{2}}{\relcomp{\red_{1}}{\mred_{2}}}}$ 
  is terminating on $B$.
  We write $\SNrel{\to_1}{\to_2}$ for relative termination on $A$.

  Let $\atrs$, $\btrs$ be TRSs over $\asig$, and $\termset \subseteq \ter{\asig,\avars}$.
  Then the TRS $\atrs$ is called 
  \emph{terminating (or strongly normalizing) relative to $\btrs$ on $\termset$}, 
  denoted $\SNrels{\atrs}{\btrs}{\termset}$, if
  $\red_{\atrs}$ is terminating relative to $\red_{\btrs}$ on $\termset$.
  We write $\SNrel{\atrs}{\btrs}$ for relative termination on all terms $\ter{\asig,\avars}$.
\end{definition}
\begin{remark}\label{rem:relterm}
  Termination of $\atrs$ relative to $\btrs$ on $\termset$
  is equivalent to:
  no term $t \in \termset$ that admits an infinite rewrite sequence
  $t = t_1 \red_{\atrs \cup \btrs} t_2 \red_{\atrs \cup \btrs} \ldots $
  containing an infinite number of $\red_{\atrs}$ steps.
  Furthermore we have $\SNrs{\atrs}{\termset}$ if and only if $\SNrels{\atrs}{\setemp}{\termset}$.
\end{remark}

\begin{definition}\normalfont\label{def:local:ralg}
  An \emph{extended well-founded monotone partial $\asig$-algebra 
  $\quadruple{\aset}{\interpret{\cdot}}{{\alggt}}{{\algge}}$}
  consists of monotone partial $\asig$-algebras 
  $\triple{\aset}{\interpret{\cdot}}{{\alggt}}$
  and
  $\triple{\aset}{\interpret{\cdot}}{{\algge}}$
  such that $\SNrel{{\alggt}}{\,{\algge}}$ holds.
\end{definition}

\noindent
Note that $\SNrel{{\alggt}}{\,{\algge}}$ implies that $\triple{\aset}{\interpret{\cdot}}{{\alggt}}$ is well-founded.
The usual condition `${\relcomp{{\alggt}}{{\algge}}} \,\subseteq\, {\alggt}$ and well-foundedness of $\alggt$'
is a special case of our condition `$\SNrel{{\alggt}}{\,{\algge}}$':
\begin{lemma}\label{lem:compat}
  Let $A$ be a set and ${\algge},\, {\alggt} \subseteq A \times A$ binary relations such that $\alggt$ is well-founded.
  Then ${\relcomp{{\alggt}}{{\algge}}} \,\subseteq\, {\alggt}$
  implies
  $\SNrel{{\alggt}}{\,{\algge}}$.
\end{lemma}
\begin{proof}
  Assume that $\SNrel{{\alggt}}{\,{\algge}}$ would not hold.
  Then there exists an infinite $(\alggt \cup \algge)$-sequence
  containing infinitely many $\alggt$ steps.
  Using ${\relcomp{{\alggt}}{{\algge}}} \,\subseteq\, {\alggt}$
  we can remove all intermediate $\algge$-steps giving rise 
  to an infinite $\alggt$-sequence, contradicting well-foundedness of $\alggt$.
\end{proof}

\begin{lemma}\label{lem:decr}
Let $\quadruple{\aset}{\interpret{\cdot}}{{\alggt}}{{\algge}}$ be an 
extended well-founded monotone partial $\asig$-algebra and let $\atrs$ and $\btrs$ 
be TRSs over $\asig$  such that $\alggt$ is a partial model for $\atrs$, and $\algge$ is a partial model for $\btrs$. 
Furthermore, assume for $s \in \ter{\asig,\setemp}$ that $\defd{\interpret{s}}$.
Then we have the implications:
\begin{enumerate}
  \item $s \red_{\atrs} t \implies \interpret{s} \alggt \interpret{t}$, and
  \item $s \red_{\btrs} t \implies  \interpret{s} \algge \interpret{t}$.
\end{enumerate}
\end{lemma}

\begin{proof}
The proofs of (i) and (ii) are identical, we just prove (ii). 
Let $s \red_{\btrs} t$, that is, we have 
a rule $\ell \to r \in \btrs$, substitution $\asubst$
and  context $\acontext$ such that
$s = \contextfill{\acontext}{\subst{\asubst}{\ell}}$
and
$t = \contextfill{\acontext}{\subst{\asubst}{r}}$. 
Since $\defd{\interpret{s}}$ and $\subst{\asubst}{\ell}$ is a subterm of $s$,
we also have $\defd{\interpret{\subst{\asubst}{\ell}}}$,
so $\interpreta{\ell}{\alpha} \algge \interpreta{r}{\alpha}$, as
$\algge$ is a partial model for $\btrs$. Then using 
 closedness and monotonicity of the interpretations $\interpret{f}$
 of all function symbols $f \in \asig$ we obtain $\interpret{s} \algge \interpret{t}$.
\end{proof}

We give a complete characterization of local relative termination
in terms of extended monotone partial algebras.
\begin{theorem}\label{thm:snrel}
  Let $\atrs$ and $\btrs$ be TRSs over $\asig$, 
  and $\termset \subseteq \ter{\asig,\setemp}$.
  Then $\SNrels{\atrs}{\btrs}{\termset}$ holds if and only if
  there is an extended well-founded monotone partial
  $\asig$-algebra $\aalg = \quadruple{\aset}{\interpret{\cdot}}{{\alggt}}{{\algge}}$
  such that the set $\termset$ is defined,
  $\alggt$ is a partial model for $\atrs$,
  and
  $\algge$ is a partial model for $\btrs$.
\end{theorem}

\begin{proof}
  For the `only if'-part assume that $\SNrels{\atrs}{\btrs}{\termset}$ holds.
  Let $\aalg = \quadruple{\aset}{\interpret{\cdot}}{{\alggt}}{{\algge}}$ 
  where $\aset \defdby \famil{\atrs \cup \btrs}{\termset}$
  and the interpretation of a function symbol  $f \in \asig$ is defined by
  $\funap{\interpret{f}}{t_1,\ldots,t_n} \defdby \funap{f}{t_1,\ldots,t_n}$
  if $\funap{f}{t_1,\ldots,t_n} \in \aset$,
  and $\undefd{\funap{\interpret{f}}{t_1,\ldots,t_n}}$ otherwise.
  The relations ${\algge}$ and ${\alggt}$ are defined by 
  ${\algge} \defdby {{\mred_{\atrs \cup \btrs}} \meet {(\aset \times \aset)}}$
  and ${\alggt} \defdby 
  {({\relcomp{\red_{\atrs}}{\mred_{\atrs \cup \btrs}}}) \meet {(\aset \times \aset)}}$.
  
  We verify that $\aalg$ is an extended well-founded monotone partial $\asig$-algebra.
  Suppose ${\alggt}$ would not be well-founded.
  Then there exists $t \in \famil{\atrs \cup \btrs}{\termset}$ admitting an infinite 
  $\relcomp{\red_{\atrs}}{\mred_{\atrs \cup \btrs}}$ rewrite sequence,
  contradicting $\SNrels{\atrs}{\btrs}{\termset}$.
  We have $\relcomp{{\alggt}}{{\algge}} \,\subseteq\, {\alggt}$ by definition,
  and consequently $\SNrel{\alggt}{\,\algge}$ by Lemma~\ref{lem:compat}.
  For $f \in \asig$ we show that $\interpret{f}$ is closed and monotone
  with respect to $\alggt$ (for $\algge$ the reasoning is the same).
  Consider  $s, t \in \aset$ with $s \alggt t$. %
  Whenever 
  $\defd{\funap{\interpret{f}}{\ldots,s,\ldots}}$
  we have  also $\defd{\funap{\interpret{f}}{\ldots,t,\ldots}}$
  (since the family $\famil{\atrs \cup \btrs}{\termset}$ is closed under rewriting),
  and hence $\funap{\interpret{f}}{\ldots,s,\ldots} \alggt \funap{\interpret{f}}{\ldots,t,\ldots}$
  as a consequence of the closure of rewriting %
  under contexts.
  Hence $\aalg$ is an extended well-founded monotone partial $\asig$-algebra.

  The set $\termset$ is defined, 
  since for every term $s \in \termset$ we have $\defd{\interpret{s}}$ by definition.
  It remains to be proved that $\alggt$ is a partial model for $\atrs$, and $\algge$ a partial model for $\btrs$.
  We only consider $\alggt$, as the reasoning for $\algge$ is the same.
  Let ${\ell \to r} \in \atrs$ %
  and $\alpha \funin \avars \to \aset$
  such that $\defd{\interpreta{\ell}{\alpha}}$.
  Then
  $\interpreta{\ell}{\alpha} = \subst{\alpha}{\ell} \red_{\atrs} \subst{\alpha}{r} 
            = \interpreta{r}{\alpha}$.
  Then  $\interpreta{\ell}{\alpha} \alggt \interpreta{r}{\alpha}$ because both
  $\interpreta{\ell}{\alpha} \in \aset$
  and $\interpreta{r}{\alpha} \in \aset$.

  For the `if'-part assume that
  $\aalg \defdby \quadruple{\aset}{\interpret{\cdot}}{{\alggt}}{{\algge}}$
  fulfilling the requirements of the theorem is given.
  Assume that $\SNrels{\atrs}{\btrs}{\termset}$ would not hold.
  Then there exists $t_0 \in \termset$
  which admits an infinite ${\red_{\atrs}} \cup {\red_{\btrs}}$ rewrite sequence
  $t_0 \to t_1 \to \ldots$ containing an infinite number of $\red_{\atrs}$-steps.
  By Lemma~\ref{lem:decr} this sequence then would give rise to 
   an infinite ${\alggt} \cup {\algge}$ sequence:
  $\interpret{t_0} \mathrel{({\alggt} \cup {\algge})} \interpret{t_1} 
                                      \mathrel{({\alggt} \cup {\algge})} \ldots$
  containing infinitely many $\alggt$-steps,
  contradicting $\SNrel{\alggt}{\,\algge}$.
\end{proof}

\begin{example}
  We consider a simple example to illustrate the method:%
  \begin{align*}
    \atrs &= \{a \to b\} & \btrs &= \{b \to b,\; \funap{f}{b} \to \funap{f}{a}\} & \termset &= \{a\}
  \end{align*}
  Global relative termination $\SNrel{\atrs}{\btrs}$ does not hold, e.g. not on $\funap{f}{a}$. However on $\termset$
  the rule $a \to b$ is terminating relative to the other rules.
  We can prove this using the extended well-founded monotone partial
  $\asig$-algebra $\aalg = \quadruple{\{0,1\}}{\interpret{\cdot}}{{>}}{{\ge}}$.
  The interpretations are given by:
  $\interpret{a} = 1$, $\interpret{b} = 0$ and $\undefd{\funap{\interpret{f}}{x}}$ for all $x \in \aset$.
  Then $\termset$ is defined,
  $\alggt$ is a partial model for $\atrs$ ($\interpret{a} = 1 > 0 = \interpret{b}$),
  and $\algge$ is a partial model for $\btrs$ ($\interpret{b} = 0 \ge 0 = \interpret{b}$ and $\undefd{\interpret{\funap{f}{b}}}$).
  Hence we conclude $\SNrels{\atrs}{\btrs}{\termset}$ by an application of Theorem~\ref{thm:snrel}.
\end{example}

See further Example~\ref{ex:SK} in Section~\ref{sec:qmodel} for a non-trivial example.

%% file: stepwise.tex
\noindent For termination proofs it is common practice to weaken 
the proof obligation stepwise by removing rules.
The idea is to find interpretations such that a part $\atrs' \subseteq \atrs$ of the rules
is decreasing ($\alggt$) and the remaining rules are weakly decreasing ($\algge$).
Then %
for termination of $\atrs$ it suffices to prove termination of 
the rules in the complement $\atrs \setminus \atrs'$.
We would also like to have this possibility for proofs of local termination.
However, for local termination we cannot simply remove (and then forget about) 
the strictly decreasing rules, 
as the following example illustrates.

\begin{example}
  Consider the set $\termset = \{a\}$ in
the TRS with the following rules:
  \begin{align*}
    a &\to b & b &\to b%
  \end{align*}
  \noindent
  We define a monotone partial $\asig$-algebra $\triple{\nat}{\sinterpret}{{>}}$
  by $\interpret{a} = 1$ and $\interpret{b} = 0$.
  Then the rule $a \to b$ is decreasing ($>$ is a partial model) since $\interpret{a} > \interpret{b}$,
  and for $b \to b$ we have $\interpret{b} = \interpret{b}$.
  However, removing the strictly decreasing rule  $a \to b$ is not sound,
 since the resulting TRS is terminating on $\termset$.
\end{example}

Let us briefly elaborate on the following theorem which enables us to
 remove rules stepwise.
Assume that the goal is proving that $\atrs$ is terminating relative to $\btrs$
 on $\termset$, that is, $\SNrels{\atrs}{\btrs}{\termset}$.
We start with zero knowledge: $\SNrels{\setemp}{\atrs \cup \btrs}{\termset}$.
We search for an interpretation that makes a part $\atrs' \subseteq \atrs$ of the rules
decreasing ($\alggt$) and the remaining rules in $\atrs \cup \btrs$ weakly decreasing ($\algge$).
Then 
the rules in $\atrs'$ can only be applied finitely often:
$\SNrels{\atrs'}{((\atrs \setminus \atrs') \cup \btrs)}{\termset}$.
\emph{But how to proceed?}
As we have seen above, we cannot simply forget about the rules $\atrs'$,
but need to take into account their influence on the family 
$\famil{\atrs \cup \btrs}{\termset}$.
A possible and theoretically complete solution would be to require 
these rules to be weakly decreasing ($\algge$).
However, for practical applicability this requirement seems too strict
as it imposes heavy restrictions on the termination order.
We propose a different approach,
which allows the `removed' rules $\atrs'$ to change arbitrarily, even increase,
the interpretation of the rewritten terms,
as long as rewriting defined terms yields defined terms again.
For this purpose we introduce a relation $\algnu$ on $\aset$,
which is a partial model for the already removed rules, and
thereby guarantees that these rules preserve definedness.

\begin{theorem}\label{thm:stepwise}
  Let $\atrs$, $\atrs'$ and $\ctrs$
  be TRSs over $\asig$, and $\termset \subseteq \ter{\asig,\setemp}$ a set of terms
  such that $\SNrels{\ctrs}{(\atrs \cup \atrs')}{\termset}$ holds.
  Then $\SNrels{(\ctrs \cup \atrs')}{\atrs}{\termset}$
  holds if and only if
  there exists an extended well-founded monotone partial
  $\asig$-algebra $\aalg = \quadruple{\aset}{\interpret{\cdot}}{{\alggt}}{{\algge}}$
  and a relation $\algnu$ on $\aset$
  such that:
  \begin{enumerate}[(1)]
    \item the set $\termset$ is defined,
    \item $\triple{\aset}{\interpret{\cdot}}{{\algnu}}$
          is a monotone partial $\asig$-algebra, and
    \item
      $\alggt$, $\algge$ and $\algnu$
      are partial models for
      $\atrs'$, $\atrs$ and $\ctrs$, respectively.
  \end{enumerate}
\end{theorem}

\begin{proof}
  Straightforward extension of the proof of Theorem~\ref{thm:snrel}.
  The `only if'-part follows immediately by taking ${\algnu} \defdby {\alggt}$.
  For the `if'-part
  consider an infinite reduction $t_1 \to t_2 \to \ldots$ with $t_1 \in \termset$.
  Then since $\algnu$ is a partial model for $\ctrs$,
  we conclude $\myall{i \in \nat}{\defd{\interpret{t_i}}}$.
  Moreover, as a consequence of $\SNrels{\ctrs}{(\atrs \cup \atrs')}{\termset}$
  we can cut off the prefix of the sequence containing the finitely many $\ctrs$ steps.
\end{proof}

\begin{example}\label{ex:fac2}
  We reconsider Example~\ref{ex:fac}, and prove termination of $\atrs$ on $T$. 
  The usage of Theorem~\ref{thm:stepwise} allows for a simpler stepwise termination proof.
  In particular, for removing the rules for $\prog{fac}$
  we can employ the standard interpretation $\mul$ as $\cdot$ and $\add$ as $+$.
  Let $\aalg = \quadruple{\nat}{\sinterpret}{{>}}{{\ge}}$ where $>$ is the natural order on $\nat$,
  and $\sinterpret$ is given by:
  \begin{align*}
    \interpret{0} &\defdby 0 &
    \funap{\interpret{\fs}}{n} &\defdby n+1 &
    \funap{\interpret{\prog{fac}}}{n} &\defdby \factorial{(n+2)} \\
    \bfunap{\interpret{\mul}}{n}{m} &\defdby n\cdot m &
    \bfunap{\interpret{\add}}{n}{m} &\defdby n + m&
    \undefd{\bfunap{\interpret{\fm}}{n}{m}}
  \end{align*}
  for all $n,m \in \nat$.
  Then $>$ is a partial model for $\eqref{rule:fac:z}$ and $\eqref{rule:fac:s}$:
  \begin{align*}
    \interpreta{\funap{\prog{fac}}{0}}{\alpha} = 2 
      &> 1 = \interpreta{\funap{\fs}{0}}{\alpha}\\
    \interpreta{\funap{\prog{fac}}{\funap{\fs}{x}}}{\alpha} = \factorial{(\funap{\alpha}{x}+3)} 
      &> (\funap{\alpha}{x}+1)\factorial{(\funap{\alpha}{x}+2)} = \interpreta{\bfunap{\mul}{\funap{\prog{fac}}{x}}{\funap{\fs}{x}}}{\alpha}
  \end{align*}
  and obviously $\ge$ is a partial model for the other rules.
  Let $U_1 = \{\eqref{rule:fac:z},\eqref{rule:fac:s}\}$, and $\atrs_1 = \atrs \setminus U_1$.
  Then by Theorem~\ref{thm:stepwise} it suffices to show $\SNrels{U_1}{\atrs_1}{\termset}$ to conclude $\SNrs{\atrs}{\termset}$.

  As second step, we remove the $\mul$ rules. 
  Let $\aalg = \quadruple{\nat}{\sinterpret}{{>}}{{\ge}}$ with:
  \begin{align*}
    \interpret{0} &\defdby 0 &
    \funap{\interpret{\fs}}{n} &\defdby n+1 &
    \funap{\interpret{\prog{fac}}}{n} &\defdby n \\
    \bfunap{\interpret{\mul}}{n}{m} &\defdby (n+1)\cdot (m+1)&
    \bfunap{\interpret{\add}}{n}{m} &\defdby n + m&
    \undefd{\bfunap{\interpret{\fm}}{n}{m}}
  \end{align*}
  for all $n,m \in \nat$.
  Recall that the rules from $U_1$ have to be taken into consideration as they have an impact on the set of reachable terms
  (otherwise the set of terms $\termset$ would consist only of normal forms).
  Nevertheless, the rule $\eqref{rule:fac:s}$ from $U_1$ is not (weakly) decreasing, that is,
  $\ge$ is not a partial model for $\eqref{rule:fac:s}$
  with respect to the above interpretation:
  \begin{align*}
    \interpreta{\funap{\prog{fac}}{\funap{\fs}{x}}}{\alpha} = \funap{\alpha}{x}+1
      &\not\ge (\funap{\alpha}{x}+1) \cdot (\funap{\alpha}{x}+2) = \interpreta{\bfunap{\mul}{\funap{\prog{fac}}{x}}{\funap{\fs}{x}}}{\alpha}
  \end{align*}
  This is also not necessary. It suffices that $U_1$ is decreasing with respect to any other relation $\algnu$
  guaranteeing that all reachable terms are defined.
  For the current example we can choose the `total' relation ${\algnu} = \{ (n,m) \where n,m \in \nat \}$
  relating all pairs of natural numbers.
  Then $\algnu$ is a partial model for $U_1$, and all $\interpret{f}$ for $f \in \asig$
  are closed and monotone with respect to $\algnu$.
  The rules $\eqref{rule:mul:zr}$, $\eqref{rule:mul:zl}$, and $\eqref{rule:mul:s}$ are decreasing ($>$ is a partial model), for all $\alpha \funin \avars \to \nat$:
  \begin{align*}
    \interpreta{\bfunap{\mul}{x}{0}}{\alpha} = \funap{\alpha}{x} + 1
      &> 0 = \interpreta{0}{\alpha} \\
    \interpreta{\bfunap{\mul}{0}{y}}{\alpha} = \funap{\alpha}{y} + 1
      &> 0 = \interpreta{0}{\alpha} \\
    \interpreta{\bfunap{\mul}{x}{\funap{\fs}{y}}}{\alpha} = (\funap{\alpha}{x} + 1) \cdot (\funap{\alpha}{y} + 2)
      &>\\ (\funap{\alpha}{x} + 1) \cdot (\funap{\alpha}{y} + 1) + \funap{\alpha}{x} &= \interpreta{\bfunap{\add}{ \bfunap{\mul}{x}{y} }{x}}{\alpha}
  \end{align*}
  The remaining rules in $\atrs_1$ are weakly decreasing (that is, $\ge$ is a partial model).
  We define $U_2 = U_1 \cup \{\eqref{rule:mul:zr}, \eqref{rule:mul:zl}, \eqref{rule:mul:s}\}$,
  and let $\atrs_2 = \atrs_1 \setminus U_2$.
  Then by Theorem~\ref{thm:stepwise} $\SNrels{U_2}{\atrs_2}{\termset}$ implies $\SNrels{U_1}{\atrs_1}{\termset}$.

  Finally, we employ the algebra $\aalg = \quadruple{\nat}{\sinterpret}{{>}}{{\ge}}$ with:
  \begin{align*}
    \interpret{0} &\defdby 0 &
    \funap{\interpret{\fs}}{n} &\defdby n+1 &
    \funap{\interpret{\prog{fac}}}{n} &\defdby n \\
    \bfunap{\interpret{\mul}}{n}{m} &\defdby n+m&
    \bfunap{\interpret{\add}}{n}{m} &\defdby n + 2m&
    \undefd{\bfunap{\interpret{\fm}}{n}{m}}
  \end{align*}
  for all $n,m \in \nat$, together with ${\algnu} = \{ (n,m) \where n,m \in \nat \}$.
  Thereby 
  $>$ is a partial model for all rules from $\atrs_2$,
  and $\algnu$ is a partial model for $U_2$.
  Hence, we conclude $\SNrels{U_2}{\atrs_2}{\termset}$, and thus $\SNrs{\atrs}{\termset}$.
\end{example}

For other applications of the theorem see Examples~\ref{ex:SK} and~\ref{ex:TM} in Section~\ref{sec:qmodel}.

%% file: models.tex
\noindent In this section we describe an easy transformation from local to global termination
based on an adaptation of semantic labeling~\cite{Z95}.
For this purpose we generalise the concept of models from~\cite{Z95}
to partial models.
Whenever the language $\termset$ for which we are interested in termination
can be described by a partial model,
that is, $\termset = \{t \where \defd{\interpret{t}}\}$,
then semantic labeling allows for a simple, complete transformation
from local to global termination.
Here complete means that
the original system is locally terminating on $\termset$
if and only if
the transformed, labeled system is globally terminating.

We define a variant of semantic labeling
where each symbol is labeled by the tuple of the values of its arguments.
\begin{definition}\normalfont\label{def:local:labeling}
  Let $\asig$ be a signature, and let
  $\aalg \defdby \pair{\aset}{\sinterpret}$ be a partial $\asig$-algebra.
  For $t \in \ter{\asig,\avars}$ and $\alpha \funin \vars{t} \to \aset$
  such that $\defd{\interpreta{t}{\alpha}}$,
  the \emph{labeling $\dolAbel{\aalg}{t}{\alpha}$ of $t$ with respect to $\alpha$} 
  is defined as follows:
  \begin{align*}
     \dolAbel{\aalg}{x}{\alpha} &\defdby x \\
     \dolAbel{\aalg}{\funap{f}{t_1,\ldots,t_n}}{\alpha} &\defdby 
     \funap{\slAbel{f}{\interpreta{t_1}{\alpha},\ldots,\interpreta{t_n}{\alpha}}}{\dolAbel{\aalg}{t_1}{\alpha},\ldots,\dolAbel{\aalg}{t_n}{\alpha}}
     \punc.
  \end{align*}
  over the signature 
  $\lAbelsig{\aalg}{\asig} = \{\slAbel{f}{\lambda} 
             \where f \in \asig,\; \lambda \in \aset^{\arity{f}} \text{ such that } \defd{\funap{\interpret{f}}{\lambda}}\}$
\end{definition}

In order to obtain a complete transformation we need to restrict the models to their core,
that is, those elements that are interpretations of ground terms.
\begin{definition}\normalfont\label{def:local:core}
  Let $\aalg = \pair{\aset}{\sinterpret}$ be a partial $\asig$-algebra.
  Then the \emph{core $\core{\aalg} \subseteq \aset$ of $\aalg$}
  is the smallest set such that
  $\funap{\interpret{f}}{a_1,\ldots,a_n} \in \core{\aalg}$ 
  whenever $f \in \asig$ and $a_1,\ldots,a_n \in \core{\aalg}$ with $\defd{\funap{\interpret{f}}{a_1,\ldots,a_n}}$.
  We say that $\aalg$ is \emph{core} if $\aset = \core{\aalg}$.
\end{definition}
By construction of the core we have $\core{\aalg} = \{\interpret{t} \where t \in \ter{\asig,\setemp},\, \defd{\interpret{t}}\}$.
The restriction of a model to its core does not change its language, 
thus in the sequel we can without loss of generality assume that all models are core.

We have arrived at the transformation from local to global termination.
The rules are labeled as known from semantic labeling
with the exception that labeled rules are thrown away if the interpretation of their left-hand side is undefined.
\begin{definition}\normalfont\label{def:local:transform}
  Let $\atrs$ be a TRS over $\asig$,
  and $\aalg = \pair{\aset}{\sinterpret}$ a partial $\asig$-algebra.
  We define the \emph{labeling of $\atrs$} 
  as the TRS $\lAbeltrs{\aalg}{\atrs}$ over the signature $\lAbelsig{\aalg}{\asig}$ by:
  \begin{align*}
    \lAbeltrs{\aalg}{\atrs} \defdby \{\dolAbel{\aalg}{\ell}{\alpha} \to \dolAbel{\aalg}{r}{\alpha} 
    \where \ell \to r \in \atrs,\; 
           \alpha \funin \vars{\ell} \to \aset \text{ such that } \defd{\interpreta{\ell}{\alpha}}\}
   \punc.
  \end{align*}
\end{definition}
A TRS is \emph{collapsing} if it contains rules of the form $\ell \to x$ with $x \in \avars$.
Such collapsing rules can be eliminated
by replacing them with all instances $\subst{\asubst_{\!f}}{\ell} \to \subst{\asubst_{\!f}}{x}$
for every $f \in \asig$
where $\funap{\asubst_{\!f}}{x} = \funap{f}{x_1,\ldots, x_n}$ with $x_1$, \ldots, $x_n$ pairwise different, fresh variables.

\begin{theorem}\label{thm:transform}
  Let $\atrs$ be a non-collapsing TRS over $\asig$,
  and $\aalg = \pair{\aset}{\sinterpret}$
  a core partial model for $\atrs$.
  Then $\atrs$ is locally terminating on $\lang{\aalg}$ if and only if $\lAbeltrs{\aalg}{\atrs}$ is globally terminating.
\end{theorem}
\begin{proof}
  We introduce types for $\lAbeltrs{\aalg}{\atrs}$ over the sorts $\aset$.
  For every symbol
  $f^{\lambda} \in \lAbelsig{\aalg}{\asig}$ with $\lambda = \tuple{a_1,\ldots,a_{\arity{f}}}$
  we define $f^\lambda$
  to have input sorts $\tuple{a_1,\ldots,a_n}$ 
  and output sort $\funap{\interpret{f}}{a_1,\ldots,a_n}$.
  Then~\cite[Proposition~5.5.24]{ohle:02}
  with non-collapsingness of $\lAbeltrs{\aalg}{\atrs}$
  yields that $\lAbeltrs{\aalg}{\atrs}$
  is terminating if and only if all well-sorted terms are terminating.
  Since $\aalg$ is core
  there exists a well-sorted ground term for every sort in $\aset$.
  Thus by application of a ground substitution we can assume that all rewrite sequences contain only ground terms,
  and the set of well-sorted ground terms is exactly the language $\lang{\aalg}$ of the model $\aalg$.
\end{proof}

To apply Theorem~\ref{thm:transform} for proving local termination of $\atrs$ on a set of terms $\termset$
we have to find a partial model $\aalg$ for $\atrs$ such that $\termset \subseteq \lang{\aalg}$.
Then global termination of $\lAbeltrs{\aalg}{\atrs}$ implies local termination of $\atrs$ on $\termset$.
If moreover we have $\family{\termset} = \lang{\aalg}$, then 
the transformation is complete, that is, the converse implication holds as well.

\begin{example}\label{ex:Smodel}
  We revisit Example~\ref{ex:S} on the $\combS$ combinator with $\termset = \{\combS^n \where n \in \nat\}$.
  We choose the partial model $\aalg = \pair{\aset}{\sinterpret}$,
  where $\aset = \{0,1,2\}$ and the interpretation is defined by:
  $\interpret{\combS} = 0$,
  $\bfunap{\interpret{\scombAt}}{0}{0} = 1$,
  $\bfunap{\interpret{\scombAt}}{1}{x} = 2$ for all $x \in \aset$,
  $\bfunap{\interpret{\scombAt}}{2}{0} = 2$,
  and $\undefd{}$ otherwise.
  Then $\termset \subseteq \lang{\aalg}$ and a short proof even shows that $\family{\termset} = \lang{\aalg}$.
  The labeling $\lAbeltrs{\aalg}{\{\combS xyz \to xz(yz)\}}$ is:
  \begin{align*}
    \bfunap{\slAbel{\scombAt}{2,0}}{\bfunap{\slAbel{\scombAt}{1,0}}{\bfunap{\slAbel{\scombAt}{0,0}}{\combS}{x}}{y}}{z} 
    &\to \bfunap{\slAbel{\scombAt}{1,1}}{\bfunap{\slAbel{\scombAt}{0,0}}{x}{z}}{\bfunap{\slAbel{\scombAt}{0,0}}{y}{z}}
    \\
    \bfunap{\slAbel{\scombAt}{2,0}}{\bfunap{\slAbel{\scombAt}{1,1}}{\bfunap{\slAbel{\scombAt}{0,0}}{\combS}{x}}{y}}{z} 
    &\to \bfunap{\slAbel{\scombAt}{1,2}}{\bfunap{\slAbel{\scombAt}{0,0}}{x}{z}}{\bfunap{\slAbel{\scombAt}{1,0}}{y}{z}}
    \\
    \bfunap{\slAbel{\scombAt}{2,0}}{\bfunap{\slAbel{\scombAt}{1,2}}{\bfunap{\slAbel{\scombAt}{0,0}}{\combS}{x}}{y}}{z} 
    &\to \bfunap{\slAbel{\scombAt}{1,2}}{\bfunap{\slAbel{\scombAt}{0,0}}{x}{z}}{\bfunap{\slAbel{\scombAt}{2,0}}{y}{z}}
    \punc.
  \end{align*}
  The other labeled rules are thrown out as their left-hand side is undefined.
  Global termination of the transformed system can be shown by the recursive path order~\cite{D82}.
\end{example}

%% file: cl-s.tex
\noindent In this section we consider TRSs with the property that
strong and weak normalisation coincide. In case the language of
normalising terms happens to be regular, we show how a tree automaton
(partial model) can be found accepting exactly the (closed)
normalising terms. We automated this procedure.  Then we label the TRS
with the obtained partial model, and employ
Theorem~\ref{thm:transform} to transform the local termination problem
for the set of normalising terms to global termination of the labelled
TRS.

Since in orthogonal, non-erasing term rewriting systems strong and weak normalisation coincide, these form a typical area where the method can be applied.
In this section we illustrate this construction with two well-known examples
from combinatory logic (CL)~\cite{cl}. We use Smullyan's bird nicknames of
the combinators~\cite{smullyan:90}.
\begin{enumerate}[(1)]
 \item The Owl, corresponding to the rewrite rule:
  \begin{align*}
    \combOwl\;xy \to y(xy)
  \end{align*}
 \item The Starling
  \begin{align*}
    \combS\;xyz \to xz(yz),
  \end{align*}
  also known as the fragment CL(S) of combinatory logic consisting of all terms solely built from application and the $\combS$-combinator.
\end{enumerate}
The termination problem of Smullyan's Owl has been solved in~\cite{jw:fpc}. Here, it serves as illustrating example.

The termination problem of CL(S) is non-trivial, and its word problem is still open.
In~\cite{DBLP:journals/iandc/Waldmann00} decidability of strong normalisation of terms in CL(S) has been shown,
and we are aiming at a formal verification of the following proposition:
\begin{proposition}[\cite{DBLP:journals/iandc/Waldmann00}]\label{prop:s}
  The set of normalising ground $\combS$-terms is a rational language.
\end{proposition}

We now turn to the construction of the partial models.

\begin{definition}
  For a tree language $L$,
  its \emph{Nerode congruence} $\sim_L$ 
  is the relation on ground terms given by
  $t_1 \sim_L t_2 \iff \forall \text{ground }C[]: C[t_1] \in L \iff C[t_2]\in L$.
\end{definition}

The next lemma follows easily by considering the Nerode congruence~\cite{tata:07}.
\begin{lemma}\label{lem:nerode}
  If a TRS $R$ has the property that
  every ground term is weakly normalising if and only if it is strongly normalising,
  and $N$ is the language of normalising ground terms, then
  \begin{enumerate}[$\bullet$]
  \item each congruence class of $\sim_N$ 
    is closed under $R$-rewriting and closed $R$-expansion,
  \item the complement of $N$ occurs 
    as one of the $\sim_N$ congruence classes.  
  \end{enumerate}
\end{lemma}
\begin{proof}
  Note that under the assumptions on $R$,
  for each term $t\in N$ and each subterm $s$ of a term in $N$ we have $s\in N$.
  In other words, $s\notin N$ implies $C[s]\notin N$.
  Also, for all ground terms $t_1$, $t_2$ with $t_1 \to_R t_2$ 
  we have $t_1\in  N \iff t_2\in N$.
  The claims follow.
\end{proof}

In particular weak and strong normalisation of terms coincide
for every orthogonal and non-erasing TRS; this applies for CL(S) as well as Smullyan's Owl.
We note that the set of congruence classes of $\sim_N$
can in general be infinite, even for orthogonal, non-erasing TRSs.

\begin{example}
  We consider an example of an orthogonal, non-erasing TRS where the set of congruence classes of $\sim_N$ is infinite.
  Let $\atrs$ consist of the rules:
  \begin{align*}
    \funap{a}{\funap{b}{x}} &\to x & \funap{c}{\funap{c}{d}} &\to \funap{c}{\funap{c}{d}}
  \end{align*}
  over the signature $\asig = \{a,b,c,d\}$ with $d$ a constant.
  Here, terms of the from $\funap{c}{\funap{a^n}{\funap{b^n}{\funap{c}{d}}}}$
  are non-terminating, while all terms of the form $\funap{c}{\funap{a^n}{\funap{b^{m}}{\funap{c}{d}}}}$
  with $n \ne m$ are terminating.
  Hence none of the terms $\funap{b^n}{\funap{c}{d}}$ for $n \in \nat$
  can be in the same congruence class of $\sim_N$.
\end{example}

\begin{corollary}
  If the set of congruence classes of $\sim_N$ is finite,
  then the minimal complete deterministic bottom-up tree automaton
  for $N$ is finite and a model for R.%
\end{corollary}
Thus, if the set of congruence classes of $\sim_N$ is finite, then the set of normalising terms is a regular language.
Assume that we are lucky and the set of congruence classes is finite.
\emph{How can we find the regular automaton accepting the set of normalising terms?}

A manual analysis and construction of the automaton as in~\cite{DBLP:journals/iandc/Waldmann00} can be tedious and error-prone.
The reference contains a hand-made tree grammar
(top-down non-de\-ter\-min\-is\-tic tree automaton $A$) and claims:
\begin{enumerate}[$\bullet$]
\item the $S$ rule is locally terminating on $L(A)$,
\item $L(A)$ contains all normal forms,
\item $A$ is closed under inverse application of the $S$ rule.
\end{enumerate}
\noindent
Starting from that grammar, we can indeed compute
a bottom-up minimal deterministic~tree automaton $B$ with $L(B)=L(A)$
(strangely, it has 39 states, and not 43, as claimed in the reference).

We propose a different, automatable approach for finding a regular automaton accepting the language of all normalising terms.
The idea is to employ the definition of the Nerode congruence for `guessing' the congruence classes.
Here we use the word `guess' in place of `compute' since we
need to check whether a term $\contextfill{C}{s}$ is terminating.
This property is in general undecidable. We can, however, make an educated guess by choosing a large enough $d$
and checking whether $\contextfill{C}{s}$ admits a rewrite sequence of length $d$ with respect to some strategy $\leadsto$.
Note that the strategy $\leadsto$ can be chosen arbitrarily since we assume that weak and strong normalisation coincide.
\begin{definition}\cite{terese:03}
  A \emph{strategy $\leadsto$} for a TRS $R$ is a relation ${\leadsto} \subseteq {\to_R}$
  on $\ter{\asig,\avars}$ having the same normal forms as $\to_R$.
  A strategy $\leadsto$ is called \emph{deterministic} if every term $t$ has at most one reduct $s$, that is, $t \leadsto s$.
\end{definition}

The following algorithm
searches for a partial model for the language of normalising terms.
The algorithm depends on 
a strategy $\leadsto$ for $R$ and
parameters $c, d \in \nat$ 
where 
$c$ is the maximal depth of contexts~$C$,
and
$d$ is the length of $\leadsto$-reductions used to guess whether a term is normalising.

Instead of the full Nerode congruence of the set of $R$-normalizing terms,
which might be undecidable,
we use an decidable equivalence relation $\sim$ on terms 
given by $s\sim t$ iff for each context $C[]$ of height $\le c$,
either both $C[s]$ and $C[t]$ have a $\leadsto$-derivation of length at least $d$,
or both don't.

\begin{algorithm}\label{alg:nerode}
  Starting from the full relation $\sim_0$ that relates all pairs of terms,
  we compute successive refinements ${\sim_0} \supset {\sim_1} \supset {\ldots}$.
  Each $\sim_i$ is given as a finite set of representatives 
  $T_i=\{t_{i,1},\ldots,t_{i,|T_i|}\}$
  with $k \neq l \Rightarrow t_{i,k}\not\sim t_{i,l}$,
  where $\sim$ is the Nerode-like relation defined above.
  The $\sim_i$-equivalence class of non-normalizing terms 
  is not explicitly represented. 
  By Lemma~\ref{lem:nerode}, this is no loss of information.
  The full relation $\sim_0$ is given by $T_0=\emptyset$.
  We set $T_{i+1} = T_{i} \cup \{s\}$ 
  where $s =f(s_1,\ldots,s_a)$, 
  for a choice of function symbol $f$ of arity $a$,
  and terms $(s_1,\ldots,s_a)\in T_i^a$,
  such that the $\leadsto$-derivation of $s$ has length $<d$,
  and $s \not\sim t$ for each $t\in T_i$.
  The algorithm stops if no such $s$ can be found.
\end{algorithm}

We remark %
that $T_1$ consists of one element,
which is a term containing a nullary symbol only.
Each $\sim_i$ computed by this algorithm constitutes a partial algebra
(on the carrier $T_i$),
since each step of the algorithm defines one part of
the interpretation of a function symbol.

The goal is that the output algebra is a partial model for $R$,
exactly capturing the Nerode congruence. 
This may fail, for two reasons.
If $d$ is chosen too small, then a normalising term $\contextfill{C}{s}$ may
mistakenly be considered non-normalising.
If $c$ is too small, then terms may accidentally be identified, although
they behave differently when put into larger contexts.

Nevertheless, it can be shown that if the language of normalising terms
is regular, then there exist appropriate parameters $c$ and $d$ such that
the algorithm will compute the correct partial model, see Lemma~\ref{lem:cd}.
The case of having chosen $c$ or $d$ too small can be detected after running the algorithm as follows.
Let $\aalg$ be the algebra computed by the algorithm.
It can be effectively checked whether $\aalg$ is a partial model for $R$,
and whether all undefined terms $\undefd{\interpret{t}}$ contain a redex with respect to $R$.
Then it automatically follows that all undefined terms are non-normalising.
Finally we can employ Theorem~\ref{thm:transform} to
transform the termination problem 
for all defined terms $\defd{\interpret{t}}$ 
into an equivalent global termination problem of $\lAbeltrs{\aalg}{\atrs}$.
If we find a termination proof for $\lAbeltrs{\aalg}{\atrs}$, then
the partial model $\aalg$ is correct and accepts exactly the language of normalising terms.
\begin{lemma}\label{lem:cd}
  Let $R$ be a TRS such that every ground term is weakly normalising if and only if it is strongly normalising.
  If the language $N$ of normalising terms
  is regular, then there exist appropriate parameters $c$ and $d$ such that
  Algorithm~\ref{alg:nerode} computes the correct partial model accepting exactly $N$.
\end{lemma}
\begin{proof}
  If the languages $N$ is regular, then the set of congruence classes of $\sim_N$ is finite.
  Let $n = \setsize{\ter{\asig,\setemp}/_{\sim_N}}$,
  and let $T \subseteq \ter{\asig,\setemp}$ be the set of all ground terms of height $\le n + 1$.
  Then $T$ contains at least one representative $t_D \in T$ for every $D \in \ter{\asig,\setemp}/_{\sim_N}$.
  For every pair $\pair{t_{D_1}}{t_{D_2}}$ of representatives with $t_{D_1} \ne t_{D_2}$
  we pick a `discriminating' context $\acontext$ such that $\contextfill{\acontext}{t_{D_1}} \in N \not\Leftrightarrow \contextfill{\acontext}{t_{D_2}} \in N $.
  Let $\mathcal{C}$ be the set of these (finitely many) contexts.
  We choose for $c$ the maximal depth of all contexts in $\acontext \in \mathcal{C}$,
  and for $d$ the maximal length of a $\leadsto$-reduction of all normalising $\contextfill{\acontext}{t} \in N$
  with $\acontext \in \mathcal{C}$ and $t \in T$.
  Then the choice of $d$ guarantees that terms $\contextfill{\acontext}{t}$ will
  not accidentally be identified as non-terminating,
  and the choice of $c$ guarantees that all non-equivalent terms in $T_i$ will be distinguished.
\end{proof}
We have implemented Algorithm~\ref{alg:nerode}; the Haskell source can be downloaded from:
\begin{center}
  \uRl{http://infinity.few.vu.nl/local/}
\end{center}
We have applied the algorithm on Smullyan's Owl and CL(S),
obtaining in both cases the minimal partial algebra
accepting the language of all normalising terms.
Further details, including the respective partial models, are given below.

\begin{example}[\textbf{Smullyan's Owl}]\mbox{} \label{owl}
Smullyan's Owl serves as illustrating example.
The set of normalising Owl-terms has been found in~\cite{jw:fpc}.
The Owl corresponds to the following rewrite rule:
\begin{align*}
  \combOwl\;xy \to y(xy)
\end{align*}
or, equivalently, in first order notation:
\begin{align}
  \combAt{\combAt{\combOwl}{x}}{y} \to \combAt{y}{\combAt{x}{y}} \label{rule:owl}
\end{align}
Applied to Rule~\eqref{rule:owl}, Algorithm~\ref{alg:nerode} computes the partial model $\aalg = \pair{\{0,1\}}{\sinterpret}$
where $\interpret{\combOwl} = 0$, and 
the interpretation of $\scombAt$ is given in Table~\ref{tab:owl}.

\begin{table}[h!]
\begin{center}
  \begin{tabular}{|c|c|c|c|}
    \hline
      & 0 & 1\\
    \hline
    0 & 1 & 1\\
    \hline
    1 & 1 & -\\
    \hline
  \end{tabular}
  \caption{Interpretation $\funap{\interpret{\scombAt}}{x,y}$ for the Owl with $x$ on the left and $y$ on the top.}
  \label{tab:owl}
\end{center}
\end{table}

\noindent Examples for terms $t \in \lang{\aalg}$, that is, normalising terms, 
are 
\begin{align*}
  \combOwl \combOwl \combOwl \ldots \combOwl, && \combOwl (\combOwl\combOwl\ldots \combOwl) \combOwl\ldots \combOwl\text{, and}
  && \combOwl (\combOwl (\combOwl \combOwl) \combOwl \combOwl) \combOwl \combOwl \combOwl
\end{align*}
For an example of a non-normalising term take $\combOwl \combOwl (\combOwl \combOwl)$.
In words, the set of undefined (non-normalising) terms can be described as follows:
a term is undefined if it contains two distinct occurrences of $\combOwl \combOwl$.
Note that the term $\combOwl \combOwl \combOwl \ldots \combOwl = (\ldots ((\combOwl \combOwl) \combOwl) \ldots) \combOwl$
contains only one occurrence of $\combOwl\combOwl$ (or $\bfunap{\scombAt}{\combOwl}{\combOwl}$ in first-order notation).

First, we check that $\aalg$ is a partial model for $R$:
\begin{align*}
  \interpreta{\combOwl\;xy}{\alpha} &= 1 = \interpreta{y(xy)}{\alpha} &&\text{ for $\funap{\alpha}{x} = 0$, $\funap{\alpha}{y} = 0$}\\
  \undefd{\interpreta{\combOwl\;xy}{\alpha}} & \text{ and } \undefd{\interpreta{y(xy)}{\alpha}} &&\text{ for $\funap{\alpha}{x} = 0$, $\funap{\alpha}{y} = 1$}\\
  \interpreta{\combOwl\;xy}{\alpha} &= 1 = \interpreta{y(xy)}{\alpha} &&\text{ for $\funap{\alpha}{x} = 1$, $\funap{\alpha}{y} = 0$}\\
  \undefd{\interpreta{\combOwl\;xy}{\alpha}} & \text{ and } \undefd{\interpreta{y(xy)}{\alpha}} &&\text{ for $\funap{\alpha}{x} = 1$, $\funap{\alpha}{y} = 1$}
\end{align*}
Second, we use induction on the term structure to show that every undefined term contains a redex.
Let $t \in \ter{\asig,\setemp}$ be a term such that $\undefd{\interpret{t}}$.
Then by definition of $\sinterpret{}$ the term $t$ is of the form $t = \combAt{t_1}{t_2}$
and either $\interpret{t_1} = \interpret{t_2} = 1$, or $\undefd{\interpret{t_1}}$, or $\undefd{\interpret{t_2}}$.
In the latter two cases it suffices to apply the induction hypothesis to $t_1$ or $t_2$, respectively.
Thus, let~$\interpret{t_1} = \interpret{t_2} = 1$.
We use induction on the term structure of $t_1$.
Again, by definition of $\sinterpret{}$ the term $t_1$ is of the form $t_1 = \combAt{t_1'}{t_2'}$
with $\interpret{t_1'} = 0$, or $\interpret{t_1'} = 1$.
If $\interpret{t_1'} = 0$, then $t_1' = \combOwl$ and $t = \bfunap{\scombAt}{\bfunap{\scombAt}{\combOwl}{t_2'}}{t_2}$, and hence $t$ contains a redex.
For $\interpret{t_1'} = 1$ we finish by applying the second induction hypothesis.

Third, we prove termination for all defined terms.
An application of Theorem~\ref{thm:transform} yields the following labelled TRS:
\begin{align*}
  \bfunap{\slabel{\scombAt}{1,0}}{\bfunap{\slabel{\scombAt}{0,0}}{\combOwl}{x}}{y} 
    &\to \bfunap{\slabel{\scombAt}{0,1}}{y}{\bfunap{\slabel{\scombAt}{0,0}}{x}{y}}
    &&\text{ for $\funap{\alpha}{x} = 0$, $\funap{\alpha}{y} = 0$}\\
  \bfunap{\slabel{\scombAt}{1,0}}{\bfunap{\slabel{\scombAt}{0,1}}{\combOwl}{x}}{y} 
    &\to \bfunap{\slabel{\scombAt}{0,1}}{y}{\bfunap{\slabel{\scombAt}{1,0}}{x}{y}}
    &&\text{ for $\funap{\alpha}{x} = 1$, $\funap{\alpha}{y} = 0$}
\end{align*}
Termination of this system can easily be proven;
for example AProVE~\cite{Aprove04} finds a termination proof using the recursive path order.

Thus, indeed, $\aalg$ is a partial model accepting exactly the normalising $\combOwl$-terms.
\end{example}

\begin{example}[\textbf{The set of normalising S-terms}]\label{starling}
For CL(S), Algorithm~\ref{alg:nerode} returns the partial model $\aalg = \pair{\aset}{\sinterpret}$
where $\aset = \{0,1,\ldots,37\}$, $\interpret{\combS} = 4$, and 
the interpretation of $\scombAt$ is given in Table~\ref{tab:s}.
Indeed, it can be checked that $\aalg$ is equivalent
to the grammar given in \cite{DBLP:journals/iandc/Waldmann00}.

\begin{table}[h!]
  \scalebox{.5}{\input{automaton}}
  \caption{Transition table for $\interpret{\scombAt}(x,y)$ with $x$ left, $y$ top.}
  \label{tab:s}
\end{table}

\noindent We have formally verified (using the proof assistant Coq~\cite{coq}) 
that $\aalg$ is a partial model for CL(S) and that the language of $\aalg$ contains all normalising terms.

For proving that CL(S) is terminating on the language of $\aalg$
we have transformed the local into a global termination problem using Definition~\ref{def:local:transform}.
The resulting TRS contains 1800 rules which are globally terminating,
as can be shown using the DP transformation with SCC decomposition~\cite{AG00}
together with simple projections and the subterm \mbox{criterion~\cite{HM04}}.
The termination proof can be found automatically, and formally verified using
the current versions of CeTA~(1.05)~\cite{ceta} and TTT2~(1.0)~\cite{TTT09}.

\end{example}

%% file: automaton.tex
\begin{tabular}{|r|r|r|r|r|r|r|r|r|r|r|r|r|r|r|r|r|r|r|r|r|r|r|r|r|r|r|r|r|r|r|r|r|r|r|r|r|r|r|}\hline& 0& 1& 2& 3& 4& 5& 6& 7& 8& 9& 10& 11& 12& 13& 14& 15& 16& 17& 18& 19& 20& 21& 22& 23& 24& 25& 26& 27& 28& 29& 30& 31& 32& 33& 34& 35& 36& 37\\ \hline 0& 36& 36& 36& 27& 1& 36& 37& 37& 37& 37& 37& 37& 37& 37& 37& 37& 37& 37& 37& 37& 37& 37& 37& 37& 37& 37& 37& 37& 37& 37& 37& 37& 37& 37& 37& 37& 37& 37\\ \hline 1& 34& 35& 35& 32& 2& 36& 37& 37& 37& 36& 37& 36& 36& 37& 37& 37& 37& 37& 37& 37& 37& 37& 37& 37& 37& 37& 37& 37& 37& 37& 37& 37& 37& 37& 37& 37& 37& 37\\ \hline 2& 36& 36& 36& 33& 2& 37& 37& 37& 37& 37& 37& 37& 37& 37& 37& 37& 37& 37& 37& 37& 37& 37& 37& 37& 37& 37& 37& 37& 37& 37& 37& 37& 37& 37& 37& 37& 37& 37\\ \hline 3& 2& 2& 2& 1& 0& 9& 8& 10& 8& 9& 10& 12& 12& 14& 14& 16& 16& 24& 18& 19& 20& 21& 22& 26& 24& 25& 26& 27& 28& 29& 30& 31& 32& 33& 34& 35& 36& 37\\ \hline 4& 6& 7& 7& 5& 3& 11& 18& 18& 18& 13& 18& 15& 15& 18& 18& 18& 18& 28& 28& 23& 26& 26& 28& 28& 28& 28& 28& 28& 28& 28& 28& 28& 28& 28& 28& 28& 28& 28\\ \hline 5& 32& 32& 32& 17& 19& 27& 33& 33& 33& 33& 33& 33& 33& 33& 33& 33& 33& 33& 33& 33& 31& 31& 33& 33& 33& 33& 33& 33& 33& 31& 31& 31& 33& 33& 36& 36& 36& 37\\ \hline 6& 34& 34& 34& 29& 20& 36& 36& 36& 36& 36& 36& 36& 36& 36& 36& 36& 36& 36& 36& 36& 36& 36& 36& 36& 36& 36& 36& 36& 36& 36& 36& 36& 36& 36& 36& 36& 36& 37\\ \hline 7& 35& 35& 35& 30& 21& 36& 36& 36& 36& 36& 36& 36& 36& 36& 36& 36& 36& 36& 36& 36& 36& 36& 36& 36& 36& 36& 36& 36& 36& 36& 36& 36& 36& 36& 36& 36& 36& 37\\ \hline 8& 36& 36& 36& 31& 20& 37& 37& 37& 37& 37& 37& 37& 37& 37& 37& 37& 37& 37& 37& 37& 37& 37& 37& 37& 37& 37& 37& 37& 37& 37& 37& 37& 37& 37& 37& 37& 37& 37\\ \hline 9& 36& 36& 36& 33& 19& 37& 37& 37& 37& 37& 37& 37& 37& 37& 37& 37& 37& 37& 37& 37& 37& 37& 37& 37& 37& 37& 37& 37& 37& 37& 37& 37& 37& 37& 37& 37& 37& 37\\ \hline 10& 36& 36& 36& 31& 21& 37& 37& 37& 37& 37& 37& 37& 37& 37& 37& 37& 37& 37& 37& 37& 37& 37& 37& 37& 37& 37& 37& 37& 37& 37& 37& 37& 37& 37& 37& 37& 37& 37\\ \hline 11& 36& 36& 36& 22& 25& 36& 37& 37& 37& 36& 37& 36& 36& 37& 37& 37& 37& 37& 37& 36& 36& 36& 37& 37& 37& 37& 37& 37& 37& 37& 37& 37& 37& 37& 37& 37& 37& 37\\ \hline 12& 36& 36& 36& 33& 25& 37& 37& 37& 37& 37& 37& 37& 37& 37& 37& 37& 37& 37& 37& 37& 37& 37& 37& 37& 37& 37& 37& 37& 37& 37& 37& 37& 37& 37& 37& 37& 37& 37\\ \hline 13& 36& 36& 36& 31& 30& 36& 37& 37& 37& 36& 37& 36& 36& 37& 37& 37& 37& 37& 37& 36& 36& 36& 37& 37& 37& 37& 37& 37& 37& 37& 37& 37& 37& 37& 37& 37& 37& 37\\ \hline 14& 36& 36& 36& 31& 30& 37& 37& 37& 37& 37& 37& 37& 37& 37& 37& 37& 37& 37& 37& 37& 37& 37& 37& 37& 37& 37& 37& 37& 37& 37& 37& 37& 37& 37& 37& 37& 37& 37\\ \hline 15& 36& 36& 36& 31& 31& 36& 37& 37& 37& 36& 37& 36& 36& 37& 37& 37& 37& 37& 37& 36& 36& 36& 37& 37& 37& 37& 37& 37& 37& 37& 37& 37& 37& 37& 37& 37& 37& 37\\ \hline 16& 36& 36& 36& 31& 31& 37& 37& 37& 37& 37& 37& 37& 37& 37& 37& 37& 37& 37& 37& 37& 37& 37& 37& 37& 37& 37& 37& 37& 37& 37& 37& 37& 37& 37& 37& 37& 37& 37\\ \hline 17& 36& 36& 36& 35& 36& 37& 37& 37& 37& 37& 37& 37& 37& 37& 37& 37& 37&& 37& 37&&&& 37&& 37& 37&& 37&&&&&&&&&\\ \hline 18& 37& 37& 37& 37& 36& 37& 37& 37& 37& 37& 37& 37& 37& 37& 37& 37& 37& 37& 37& 37& 37& 37& 37& 37& 37& 37& 37& 37& 37& 37& 37& 37& 37& 37& 37& 37& 37& 37\\ \hline 19& 37& 37& 37& 36& 27& 37& 37& 37& 37& 37& 37& 37& 37& 37& 37& 37& 37& 37& 37& 37&&& 37& 37& 37& 37& 37& 37& 37&&&& 37& 37&&&&\\ \hline 20&&&& 37& 32&&&&&&&&&&&&&&&&&&&&&&&&&&&&&&&&&\\ \hline 21&&&& 37& 33&&&&&&&&&&&&&&&&&&&&&&&&&&&&&&&&&\\ \hline 22& 37& 37& 37& 36& 37& 37& 37& 37& 37& 37& 37& 37& 37& 37& 37& 37& 37&& 37&&&&& 37&&& 37&& 37&&&&&&&&&\\ \hline 23& 36& 36& 36& 36& 36& 36& 37& 37& 37& 36& 37& 36& 36& 37& 37& 37& 37& 37& 37& 37& 37& 37& 37& 37& 37& 37& 37& 37& 37& 37& 37& 37& 37& 37& 37& 37& 37& 37\\ \hline 24& 36& 36& 36& 36& 36& 37& 37& 37& 37& 37& 37& 37& 37& 37& 37& 37& 37&& 37& 37&&&& 37&& 37& 37&& 37&&&&&&&&&\\ \hline 25& 37& 37& 37& 37& 36& 37& 37& 37& 37& 37& 37& 37& 37& 37& 37& 37& 37& 37& 37& 37&&& 37& 37& 37& 37& 37& 37& 37&&&& 37& 37&&&&\\ \hline 26& 36& 36& 36& 36& 36& 37& 37& 37& 37& 37& 37& 37& 37& 37& 37& 37& 37& 37& 37& 37& 37& 37& 37& 37& 37& 37& 37& 37& 37& 37& 37& 37& 37& 37& 37& 37& 37& 37\\ \hline 27& 37& 37& 37& 37& 37& 37& 37& 37& 37& 37& 37& 37& 37& 37& 37& 37& 37&& 37& 37&&&& 37&& 37& 37&& 37&&&&&&&&&\\ \hline 28& 37& 37& 37& 37& 37& 37& 37& 37& 37& 37& 37& 37& 37& 37& 37& 37& 37& 37& 37& 37& 37& 37& 37& 37& 37& 37& 37& 37& 37& 37& 37& 37& 37& 37& 37& 37& 37& 37\\ \hline 29&&&& 37& 34&&&&&&&&&&&&&&&&&&&&&&&&&&&&&&&&&\\ \hline 30&&&& 37& 36&&&&&&&&&&&&&&&&&&&&&&&&&&&&&&&&&\\ \hline 31&&&& 37& 37&&&&&&&&&&&&&&&&&&&&&&&&&&&&&&&&&\\ \hline 32& 37& 37& 37& 37& 36& 37& 37& 37& 37& 37& 37& 37& 37& 37& 37& 37& 37&& 37&&&&& 37&&& 37&& 37&&&&&&&&&\\ \hline 33& 37& 37& 37& 37& 37& 37& 37& 37& 37& 37& 37& 37& 37& 37& 37& 37& 37&& 37&&&&& 37&&& 37&& 37&&&&&&&&&\\ \hline 34&&&&& 35&&&&&&&&&&&&&&&&&&&&&&&&&&&&&&&&&\\ \hline 35&&&&& 36&&&&&&&&&&&&&&&&&&&&&&&&&&&&&&&&&\\ \hline 36&&&&& 37&&&&&&&&&&&&&&&&&&&&&&&&&&&&&&&&&\\ \hline 37&&&&&&&&&&&&&&&&&&&&&&&&&&&&&&&&&&&&&&\\ \hline \end{tabular}

%% file: rfc.tex
\newcommand{\sRFC}{\operatorname{RFC}}
\newcommand{\RFC}{\funap{\sRFC}}
\newcommand{\rhs}{\funap{\operatorname{rhs}}}

\noindent We show that the method proposed in Section~\ref{sec:model}
is not only useful for local termination,
but can fruitfully be employed for global termination as well.
In~\cite{dershowitz:81}, Dershowitz reduces global termination
of right-linear TRSs to local termination on the set $\RFC{R}$, 
called the right-hand sides of forward closures of $R$.
The set $\RFC{R} \subseteq \ter{\asig,\avars}$ is a subset of all terms,
weakening the proof obligation, and often allowing for simpler termination proofs.
Previously,
the only automated method employing this transformation for proving global termination
has been the method of match-bounded string rewriting \cite{Matchbounds04}.
In the present paper we advocate an alternative approach.

We propose a combination of the $\sRFC$-method with the transformation from Section~\ref{sec:model}.
More precisely, we first reduce the global termination problem to a local termination problem on $\RFC{R}$,
and then we transform this problem back into a global termination problem.
We show that this method can successfully be applied to obtain proofs for global termination;
see further Example~\ref{ex:rfc} for a rewrite system that remained unsolved in the termination competition~\cite{term:comp:08}.

A string rewriting systems (SRS) $R$
is a TRS $R$ where all symbols $f \in \asig$  of the signature are unary.
We then use words $a_1\ldots a_n$ to denote terms $\funap{a_n}{\ldots\funap{a_1}{x}}$.

For SRSs $R$ the set $\RFC{R}$ can be defined as follows:

\begin{definition}[\cite{dershowitz:81}]
  Let $R$ be a SRS over $\Sigma$.
  The \emph{right-hand sides of forward closures of~$R$}, denoted $\RFC{R}$,
  are defined as the smallest set $F \subseteq \asig^*$ such that:
  \begin{enumerate}[$\bullet$]
    \item $\rhs{R} \subseteq F$,
    \item if $u \in F$ and $u \to v$, then $v \in F$ (rewriting), and
    \item if $u\ell_1 \in F$ and $\ell_1\ell_2 \to r \in R$ with $\ell_1 \neq \wordemp$, then $u r \in F$ (right extension).
  \end{enumerate}
\end{definition}

\noindent
We have the following well-known theorem:
\begin{theorem}[\cite{dershowitz:81}]
  A string rewriting system $R$ is terminating on $\Sigma^*$ if and only if $R$ is terminating on $\RFC{R}$.
\end{theorem}
\noindent

The set $\RFC{R}$ can be (over-)approximated by using the system 
\[
R_\# = \{ u \# \to r \# \mid (u\cdot v \to r)\in R, u\neq \epsilon, v\neq\epsilon\}
\]
over $\Sigma_\#=\Sigma\cup\{\#\}$,
where $\#$ acts as an end marker.
Then $\RFC{R}\# = \funap{(R \cup R_\#)^*}{\rhs{R}\#}$,
and we can reduce the global termination problem of $R$
to local termination of $R \cup R_\#$ on $\rhs{R}\#$.
More generally we have the following observation:
whenever $M \supseteq \rhs{R}\#$ is closed w.r.t. $R\cup R_\#$, then $\RFC{M}\# \subseteq M$.
The closure under rewriting can be proven by giving a partial model $\aalg$
($M$ is the language of a partial $\Sigma_\#$-algebra $\aalg$).

\begin{example}\label{ex:rfc}
Take $\Sigma=\{a,b,c\}$ and 
\[
R=\{a\to\epsilon, b\to\epsilon,cc\to a, ba\to cacbb\}.
\]
This is the mirrored version of \uRl{SRS/Waldmann07b/size-12-alpha-3-num-223}
which has not been solved automatically in previous termination competitions.
We present a partial $\Sigma_\#$-algebra $\aalg$ with 3 elements $\aset = \{1,2,3\}$
and interpretations of function symbols:
\[a: 1\mapsto 1, 2\mapsto 2; b : 1\mapsto 1, 
c:1\mapsto 2, 2\mapsto 1, \#:1\mapsto 3, 2\mapsto 3,
\]
and $b(2)$ as well as all transitions from $3$ are undefined.
Note: we consider the right end of the string to be the top symbol of the term.
It can be checked that $\aalg$ is a partial model for $R\cup R_\#$,
and its language contains $\rhs{R}\#$.
As a consequence we have $\lang{\aalg} \subseteq \RFC{R}\#$.
Formally, for the existence of ground terms, we add a fresh constant $e$ 
with interpretation $\interpret{e} = 1$.
This constant does not harm the property of $\aalg$ being a partial model for $R\cup R_\#$,
and does not affect the termination behaviour:
since SRSs are linear, $R\cup R_\#$ is terminating on $\rhs{R}\#$
if and only if $R\cup R_\#$ is terminating on 
$\{\funap{a_n}{\ldots\funap{a_1}{e}} \where a_1\ldots a_n \in \rhs{R}\#\}$.

We obtain the following labelled system:
\[R_\aalg=\{ a_1\to \epsilon, a_2\to\epsilon, b_1\to \epsilon, 
c_1 c_2 \to a_1, c_2 c_1 \to a_2, b_1 a_1\to c_1 a_2 c_2 b_1 b_1 \},
\]
termination of which is equivalent to termination of $R$.
Indeed $R_\aalg$ is easily seen to be terminating.
E.g., Torpa~\cite{torpa:05} finds the following termination proof:
\begin{indentation}{.7cm}{0cm}
\begin{verbatim}
[A] Choose polynomial interpretation
      a1 c1: lambda x.x+1,
      rest identity
    remove: a1  -> 
    remove: c2 c1  -> a2 
[AC] Reverse every lhs and rhs and choose polynomial 
     interpretation:
      a1 and c1: lambda x.10x, 
      rest lambda x.x+1
    remove: a2  -> 
    remove: b1 a1  -> c1 a2 c2 b1 b1 
    remove: b1  -> 
    remove: c1 c2  -> a1 
Terminating since no rules remain.
\end{verbatim}
\end{indentation}

\noindent
For automating this method,
the challenge is to find a partial model
such that the resulting labelled total termination problem
is easier than the original one.
In particular, the domain of the partial algebra
must be a proper subset of the full algebra ($\Sigma^*$).
In our example, the domain excludes 
all words containing the factor $bcb$.
\qed
\end{example}

%% file: qmodels.tex
\noindent In Sections~\ref{sec:sn}--\ref{sec:stepwise} we have devised
a characterisation of local termination in terms of monotone partial
algebras.  While this gives the general method, for the purpose of
obtaining automatable methods we strive for fruitful classes of these
algebras.  For global termination, instances of monotone algebras are
well-known.  This raises the natural question whether we can transform
a given monotone algebra for global termination in such a way that we
obtain a partial monotone algebra for local termination.

In this section we present one such approach.
We combine monotone partial models with (ordinary) monotone algebras.
The monotone partial models are roughly deterministic tree automata
that are closed under rewriting;
they describe the language of term on which we proof termination.
We search for such an automaton that accepts the starting language $\termset$
together with a monotone algebra such that the rewrite rules decrease
on the language of the automaton.
In this way monotone algebras for global termination carry over to local termination,
and we obtain an automatable method that 
is applicable for proofs of local termination.

First we give the definition of extended $\samumap$-monotone algebras as known from 
global termination of context-sensitive TRSs, see~\cite{luca:98,EWZ06}.
A mapping $\samumap \funin \asig \to \powerset{\nat}$
is called a \emph{replacement map (for $\asig$)} if for all $f \in \asig$ we have
$\amumap{f} \subseteq \{1,\ldots,\arity{f}\}$.
Let $\pair{\aset}{\sinterpret}$ be a $\asig$-algebra and $\samumap$ a replacement map.
For symbols $f \in \asig$ we say that the interpretation $\interpret{f} \funin \aset^{\arity{f}} \to \aset$
is \emph{$\samumap$-monotone} with respect to $\alggt$
if for every $a, b \in \aset$ and $i \in \amumap{f}$ with $a \alggt b$ we have:
\(\funap{f}{\underbrace{\ldots}_{i-1},a,\underbrace{\ldots}_{\arity{f}-i}} \alggt \funap{f}{\ldots,b,\ldots}\punc.\)

\begin{definition}\normalfont\label{def:globalalg}
  Let $\samumap$ be a replacement map for $\asig$.

  \noindent An \emph{extended well-founded $\samumap$-monotone $\asig$-algebra 
  $\quadruple{\aset}{\interpret{\cdot}}{{\alggt}}{{\algge}}$}
  is a $\asig$-algebra $\pair{\aset}{\interpret{\cdot}}$
  with two binary relations $\alggt$, $\algge$ on $\aset$ for which
  the following conditions hold:
  \begin{enumerate}[(i)]
    \item $\SNrel{{\alggt}}{\,{\algge}}$, and
    \item for every $f \in \asig$ the function $\interpret{f}$ is $\samumap$-monotone with respect to $\alggt$ and $\algge$.
  \end{enumerate}
\end{definition}

A partial model $\aalg = \triple{\aset}{\sinterpret}{{\qge}}$
may contain elements $a \in \aset$
for which $\interpret{t} = a$ implies that $t$ is a normal form.
For a given partial model the set of these, which we denote by $\nf{\aset}{\atrs}$, can be computed (see below Definition~\ref{def:nfalg}).
We can exploit this knowledge as follows:
if a certain argument of a symbol $f \in \asig$
is always a normal form, then its interpretation $\interpret{f}$
does not need to be monotonic for this argument position.
The following definition gives an algorithm for computing the set $\nf{\aset}{\atrs}$.
Elements that are interpretations $\interpreta{\ell}{\alpha}$ of left-hand sides in $\atrs$
cannot belong to this set. Moreover if $a \not\in \nf{\aset}{\atrs}$
and $b = \funap{\interpret{f}}{\ldots,a,\ldots}$ then we conclude $b \not\in \nf{\aset}{\atrs}$.
This is formalised as follows:
\begin{definition}\normalfont\label{def:nfalg}
  Let $\atrs$ be a TRS over the signature $\asig$, and $\aalg = \pair{\aset}{\sinterpret}$ a partial $\asig$-algebra.
  The \emph{normal forms $\nf{\aset}{\atrs}$ of $\aalg$}
  are the largest set $\nf{\aset}{\atrs} \subseteq \core{\aalg}$ such that
  $\interpreta{\ell}{\alpha} \not\in \nf{\aset}{\atrs}$
  for every $\ell \to r \in \atrs$ and every $\alpha \funin \vars{\ell} \to \core{\aalg}$,
  and
  $\funap{\interpret{f}}{a_1,\ldots,a_n} \not\in \nf{\aset}{\atrs}$ 
  for every $f \in \asig$, $a_i \not\in \nf{\aset}{\atrs}$ and $a_1,\ldots,a_n \in \core{\aalg}$.
\end{definition}
Then by construction we obtain the following lemma:
\begin{lemma}\label{lem:nfalg}
  $\nf{\aset}{\atrs}$ consists of all $a \in \core{\aalg}$ 
  for which 
  every term $t \in \ter{\asig,\setemp}$ with $\interpret{t} = a$ is a normal form with respect to $\atrs$.
\end{lemma}
As mentioned above the interpretations do not need to be monotonic in argument positions
which are normal forms. We formalise this by defining a replacement map 
for the labelling $\lAbeltrs{\aalg}{\atrs}$ of $\atrs$ which does not contain
argument positions that are in normal form.
\begin{definition}\normalfont\label{def:qtransform}
  Let $\atrs$ be a TRS over $\asig$,
  and $\aalg = \pair{\aset}{\sinterpret}$ a partial $\asig$-algebra.
  Let the replacement map $\nfamumap{\atrs}$ be defined for 
  every symbol $f^{\lambda} \in \lAbelsig{\aalg}{\asig}$ with $\lambda = \tuple{a_1,\ldots,a_{\arity{f}}}$ as follows:
  $\funap{\nfamumap{\atrs}}{f^{\lambda}} = \{1,\ldots,\arity{f}\} \setminus \{i \where a_i \in \nf{\aset}{\atrs}\}$.
\end{definition}

As an instance of Theorem~\ref{thm:stepwise}
we obtain a method for stepwise rule removal for local termination that is based
on a combination of monotone partial models 
and extended monotone algebras.

\begin{theorem}\label{thm:quasicomb}
  Let $\atrs$, $\atrs'$ and $\ctrs$ 
  be TRSs over $\asig$, and $\termset \subseteq \ter{\asig,\setemp}$ a set of terms such that
  $\SNrels{\ctrs}{\atrs \cup \atrs'}{\termset}$ holds.
  Furthermore let 
  $\aalg = \triple{\aset}{\sinterpret}{{\qge}}$
  be a monotone partial model for $\atrs \cup \atrs' \cup \ctrs$ with $\termset \subseteq \lang{\aalg}$,
  and $\balg = \quadruple{\bset}{\interpret{\cdot}_\balg}{{\alggt}}{{\algge}}$
  an extended well-founded $\nfamumap{\atrs \cup \atrs'}$-monotone $(\lAbelsig{\aalg}{\asig})$-algebra
  such that:
  \begin{enumerate}[\em(1)]
   \item $\pair{\bset}{{\alggt}}$ is a model for $\lAbeltrs{\aalg}{\atrs'}$,
   \item $\pair{\bset}{{\algge}}$ is a model for $\lAbeltrs{\aalg}{\atrs}$, and
  \item for all $f \in \asig$, $\vec{a}_1\, a\, \vec{a}_2\in \aset^{\arity{f}}$, $a \qge a' \in \aset$, and $b_1,\ldots,b_{\arity{f}} \in \bset$:
        \[\funap{\interpret{\slAbel{f}{\vec{a}_1\, a\, \vec{a}_2}}_\balg}{b_1,\ldots,b_{\arity{f}}}
          \algge
          \funap{\interpret{\slAbel{f}{\vec{a}_1\, a'\, \vec{a}_2}}_\balg}{b_1,\ldots,b_{\arity{f}}}
          \punc.
        \]
  \end{enumerate}
  Then $\SNrels{(\ctrs \cup \atrs')}{\atrs}{\termset}$ holds.
\end{theorem}
\begin{proof}%
  We construct an
  extended well-founded monotone partial
  $\asig$-algebra $\calg = \quadruple{\cset}{\interpret{\cdot}_\calg}{{\alggt}_\calg}{{\algge}_\calg}$
  fulfilling the requirements of Theorem~\ref{thm:stepwise}.
  Let $\cset = \aset \times \bset$,
  and define 
  $\pair{a_1}{b_1} \alggt_\calg \pair{a_2}{b_2} 
   \Longleftrightarrow a_1 \not\in \nf{\aset}{\atrs \cup \atrs'} \;\&\; a_1 \qge a_2 \;\&\; b_1 \alggt b_2$
  and
  $\pair{a_1}{b_1} \algge_\calg \pair{a_2}{b_2} 
   \Longleftrightarrow a_1 \not\in \nf{\aset}{\atrs \cup \atrs'} \;\&\; a_1 \qge a_2 \;\&\; b_1 \algge b_2$.
  Note that the $\nfamumap{\atrs \cup \atrs'}$-monotonicity is implemented
  by excluding elements $\pair{a_1}{b_1}$ with $a_1 \in \nf{\aset}{\atrs \cup \atrs'}$ from being sources of $\alggt \cup \algge$ steps.
  Then for each $f \in \asig$:
  $\funap{\interpret{f}_\calg}{\pair{a_1}{b_1},\ldots,\pair{a_{\arity{f}}}{b_{\arity{f}}}}
  \!=\!\pair{\funap{\interpret{f}_\aalg}{a_1,\ldots,a_{\arity{f}}}}
         {\funap{\interpret{\slAbel{f}{a_1,\ldots,a_{\arity{f}}}}_\balg}{b_1,\ldots,b_{\arity{f}}}}$
  if $\defd{\funap{\interpret{f}_\aalg}{a_1,\ldots,a_{\arity{f}}}}$, and $\undefd{}$ otherwise.
  Finally, we define the relation $\algnu$ on $\cset$ by 
  $\pair{a_1}{b_1} \algnu \pair{a_2}{b_2} \Longleftrightarrow a_1 \qge a_2$.
  Now it is straightforward to check that all requirements of Theorem~\ref{thm:stepwise}
  are fulfilled, and we conclude $\SNrels{(\ctrs \cup \atrs')}{\atrs}{\termset}$.
\end{proof}

Let us briefly elaborate on the theorem.
As an instance of Theorem~\ref{thm:stepwise}, Theorem~\ref{thm:quasicomb}
is applicable for proving local termination as well as local relative termination.
We start without knowledge $\SNrels{\setemp}{\atrs \cup \btrs}{\termset}$ and stepwise `remove' rules,
more precisely, we move rules from the right side to the left side of the slash `$/$'.
If we reach the goal $\SNrels{\atrs}{\btrs}{\termset}$, 
then the proof has been successful.

The use of partial monotone partial models for $\atrs \cup \atrs' \cup \ctrs$ 
with $\termset \subseteq \lang{\aalg}$
guarantees that the language we consider is closed under rewriting.
The set $\atrs'$ is the set of strictly decreasing rules that we are aiming to remove.
The $\nfamumap{\atrs \cup \atrs'}$-monotone $\lAbelsig{\aalg}{\asig}$-algebra $\balg$
then has the task to make all labelled rules stemming from $\atrs'$ strictly decreasing ($\alggt$),
and from $\atrs$ weakly decreasing ($\algge$).
Then we conclude that $\atrs' \cup \ctrs$ is terminating relative to 
$\atrs$ on $\termset$.

\begin{example}[Klop, see \cite{barendregt:84}, Exercise 7.4.7]
\label{ex:SK}
Example~\ref{ex:S} can be generalised 
to include the combinator $\combK$,
which has the reduction rule $\combK x y \to x$.
The initial language of flat $\combS,\combK$-terms is $\termset = (\combS|\combK)^*$; 
for example $\combS \combS \combK \combS = (((\combS \combS) \combK) \combS)$.
The partial model presented in Example~\ref{ex:Smodel}
can be extended to a monotone partial model for this generalised example 
by fixing $\interpret{\combK} = 0$
and $2 > 0$, $2 > 1$.
Note that this is not a model due to
$\interpreta{\combK x y}{\alpha} = 2 > 0 = \interpreta{x}{\alpha}$ 
for $\alpha = \mylam{z}{0}$.
For the complete proof, employing this model, we refer to: %
\begin{center}
\uRl{{\tt http://infinity.few.vu.nl/local/}}.
\end{center}
\end{example}

The second example illustrates the stepwise rule removal.
\begin{example}\label{ex:TM}{%
  \newcommand{\free}{1}      %
  \newcommand{\sright}{R}
  \newcommand{\sleft}{L}
  \newcommand{\slleft}{F}
  \newcommand{\wall}{\Box}   %
  \newcommand{\blank}{0}     %
  \newcommand{\start}{S}
  \newcommand{\tm}{\msf{M}}
  \newcommand{\finish}{\mit{finish}}
  We use a Turing-machine-like TRS which does the following.
  Starting with its head between two symbols $\free$, the tape
  containing  a finite string of $\free$'s and further blanks ($\blank$),
  it initially puts two boxes $\wall$ left and right of its head
  and afterwards alternately runs left and right between the boxes,
  each time moving them one position further, until 
  the blanks are reached:
  \begin{align*}
    &\free\free\wall\free\sright\free\free\free\free\free\wall\free\free
    \mred \free\free\wall\free\free\free\free\free\free\sright\wall\free\free\\
    \red\ &\free\free\wall\free\free\free\free\free\free\sleft\free\wall\free
    \mred \free\free\wall\sleft\free\free\free\free\free\free\free\wall\free\\
    \red\ &\free\wall\free\sright\free\free\free\free\free\free\free\wall\free \mred \ldots
  \end{align*}
  This is implemented by the TRS $\atrs$ consisting of the following rules:
  \begin{align*}
    \free\start\free &\to \wall\sright\wall &
    \sright\free &\to \free\sright &
    \sright\wall\free &\to \sleft\free\wall &
    \sright\wall\blank &\to \slleft\wall\blank \\
    &&
    \free\sleft &\to \sleft\free &
    \free\wall\sleft &\to \wall\free\sright &
    \blank\wall\sleft &\to \blank\wall\sright \\
    &&
    \free\slleft &\to \slleft\free &
    \free\wall\slleft &\to \wall\free\sright &
    \blank\wall\slleft &\to \finish
  \end{align*}
  where all symbols apart form $\finish$ (which is a constant) are unary,
  but have been written without parenthesis for the purpose of compactness.
  Note that the construction of the TRS is similar to
 the standard translation of Turing machines
  to string rewriting systems as given in~\cite{terese:03}.

  While the Turing machine is terminating on every input,
  the TRS $\atrs$ fails to be globally terminating.
  The reason is that $\atrs$ allows for configurations with multiple heads working at the same time on the same tape:
  \begin{gather*}
   \blank\wall\sright\wall\free\slleft\wall\blank 
   \red \blank\wall\sleft\free\wall\slleft\wall\blank 
   \red \blank\wall\sleft\wall\free\sright\wall\blank
   \red^2 \blank\wall\sright\wall\free\slleft\wall\blank \to \ldots
  \end{gather*}
  We will prove that $\atrs$ is locally terminating on all terms
  containing arbitrary occurrences of the symbols $\blank$, $\free$
  and at most one occurrence of $\start$, that is,
  the language given by $\termset = \{\blank,\free\}^* \,\start\, \{\blank,\free\}^* \finish$.
  As the first step we remove the rules $\free\start\free \to \wall\sright\wall$ and $\blank\wall\slleft \to \finish$.
  We do this by using a monotone model $\aalg$ consisting of only one element, accepting all terms.
  We combine this model with the $\lAbelsig{\aalg}{\asig}$-algebra $\balg$
  where $\bset = \nat$
  and $\funap{\interpret{\free}_\balg}{x} = \funap{\interpret{\blank}_\balg}{x} = x+1$,
  all other symbols are interpreted as $\mylam{x}{x}$.
  This makes the above two rules decreasing ($\alggt$ is a model for them).

  In the second step, %
  we use a partial model $\aalg = \triple{\aset}{\sinterpret}{{\qge}}$ where $\aset = \{0,1\}$,
  $0 \ge 0$, $1 \ge 1$ (but not $1 \ge 0$),
  $\interpret{\finish} = 0$ and the other interpretations are given in Table~\ref{tab:turing}:

  \begin{table}[h!]
  \begin{center}
  { \renewcommand{\arraystretch}{1.3}
  \begin{tabular}{|@{\hspace{1ex}\extracolsep{1ex}}c@{\hspace{1ex}\extracolsep{.4ex}}|c@{\hspace{.4ex}\extracolsep{.4ex}}|c@{\hspace{.4ex}\extracolsep{.4ex}}|c@{\hspace{.4ex}\extracolsep{.4ex}}|c@{\hspace{.4ex}\extracolsep{.4ex}}|c@{\hspace{.4ex}\extracolsep{.4ex}}|c@{\hspace{.4ex}\extracolsep{.4ex}}|c@{\hspace{.4ex}\extracolsep{.4ex}}|}
    \hline
    $x$ & $\funap{\interpret{\free}}{x}$ & $\funap{\interpret{\wall}}{x}$ & $\funap{\interpret{\sright}}{x}$ & $\funap{\interpret{\sleft}}{x}$ & $\funap{\interpret{\slleft}}{x}$ & $\funap{\interpret{\blank}}{x}$ & $\funap{\interpret{\start}}{x}$ \\
    \hline
    $0$ & $0$ & $1$ & $\undefd{}$ & $\undefd{}$ & $\undefd{}$ & $0$         & $0$          \\
    \hline
    $1$ & $1$ & $0$ & $1$         & $1$         & $1$         & $\undefd{}$ & $\undefd{}$ \\
    \hline
  \end{tabular}}
  \caption{Symbol interpretations.}
  \label{tab:turing}
  \end{center}
  \end{table}

  \noindent As required by the theorem $\aalg$ is a monotone partial model for $\atrs$ including the two removed rules
  $\ctrs = \{\free\start\free \to \wall\sright\wall,\; \blank\wall\slleft \to \finish\}$
  (without them $\termset$ would consist of normal forms).
  We use this partial model together with the extended well-founded monotone $\lAbelsig{\aalg}{\asig}$-algebra 
  $\balg= \quadruple{\nat}{\interpret{\cdot}_\balg}{{\alggt}}{{\algge}}$ 
  where $\alggt$ and $\algge$ are the usual orders $>$ and $\ge$ on $\nat$, respectively.
  The interpretation $\interpret{\cdot}_\balg$ is
  $\interpret{\finish}_\balg = 7$,
  $\funap{\interpret{\slAbel{\free}{0}}_\balg}{x} = 2\cdot x + 1$,
  $\funap{\interpret{\slAbel{\free}{1}}_\balg}{x} = 2\cdot x$,
  $\funap{\interpret{\slAbel{\wall}{0}}_\balg}{x} = \funap{\interpret{\slAbel{\wall}{1}}_\balg}{x} = x$,
  $\funap{\interpret{\slAbel{\sright}{1}}_\balg}{x} = 2\cdot x$,
  $\funap{\interpret{\slAbel{\sleft}{1}}_\balg}{x} = 2\cdot x + 1$,
  $\funap{\interpret{\slAbel{\slleft}{1}}_\balg}{x} = 2\cdot x$,
  $\funap{\interpret{\slAbel{\blank}{0}}_\balg}{x} = 2\cdot x$, and
  $\funap{\interpret{\slAbel{\start}{0}}_\balg}{x} = 5\cdot x + 6$.
  Then $\atrs'$ consists of the following rules:
  $\sright\wall\free \to \sleft\free\wall$, 
  $\free\sleft \to \sleft\free$, 
  $\free\wall\sleft \to \wall\free\sright$, 
  $\blank\wall\sleft \to \blank\wall\sright$, and
  $\free\wall\slleft \to \wall\free\sright$.
  Then $\pair{\balg}{{\alggt}}$ is a model for $\lAbeltrs{\aalg}{\atrs'}$.
  For instance consider the rule $\sright\wall\free \to \sleft\free\wall$.
  The labelling
  $\slAbel{\sright}{1}\slAbel{\wall}{0}\slAbel{\free}{0} \to \slAbel{\sleft}{1}\slAbel{\free}{1}\slAbel{\wall}{0}$
  is in $\lAbeltrs{\aalg}{\atrs'}$ 
  and its interpretation in $\balg$ is:
  $\funap{\slAbel{\sright}{1}\slAbel{\wall}{0}\slAbel{\free}{0}}{x} 
   = 4\cdot x + 2 
   > 4\cdot x + 1
   = \funap{\slAbel{\sleft}{1}\slAbel{\free}{1}\slAbel{\wall}{0}}{x}$.
  The labelling
  $\slAbel{\sright}{0}\slAbel{\wall}{1}\slAbel{\free}{1} \to \slAbel{\sleft}{0}\slAbel{\free}{0}\slAbel{\wall}{1}$
  is not in $\lAbeltrs{\aalg}{\atrs'}$ since its left-hand side is undefined with respect to $\aalg$,
  thus we can ignore this rule.
  Analogously it can be verified 
  $\pair{\balg}{{\algge}}$ is a model for $\lAbeltrs{\aalg}{\atrs \setminus \atrs'}$.
  Since $>$ is the empty relation on $\aset$
  the third condition of Theorem~\ref{thm:quasicomb} holds trivially.

  The three remaining rules
  $\sright\free \to \free\sright$, $\free\slleft \to \slleft\free$, and $\sright\wall\blank \to \slleft\wall\blank$
  are even globally terminating.
  This corresponds to taking a model which has only one state and accepts all terms
  together with the corresponding termination order which proves global termination.
  Hence we have proven $\SNrs{\atrs}{\termset}$ by three consecutive applications of Theorem~\ref{thm:quasicomb}.
}\end{example}

Finally, we give a theorem that allows us to remove rules and forget about them.
We need to be sure that these rules do not influence the family, that is, the set of reachable terms.
This is guaranteed if all terms in the family are normal forms with 
respect to these rules.

\begin{theorem}\label{thm:qremove}
  Let $\atrs$, $\atrs'$ and $\btrs$
  be TRSs over $\asig$, and $\termset \subseteq \ter{\asig,\setemp}$.
  Let 
  $\aalg = \triple{\aset}{\sinterpret}{{\qge}}$
  be a monotone partial model for $\atrs \cup \atrs' \cup \btrs$ with $\termset \subseteq \lang{\aalg}$
  such that for all rules $\ell \to r \in \atrs'$ and $\alpha \funin \vars{\ell} \to \aset$
  we have $\undefd{\interpreta{\ell}{\alpha}}$ (the left-hand side is undefined).
  Then $\SNrels{\atrs}{\btrs}{\termset}$ implies $\SNrels{\atrs \cup \atrs'}{\btrs}{\termset}$.
\end{theorem}
\begin{proof}
  From $\famil{\atrs \cup \atrs' \cup \btrs}{\termset} \subseteq \lang{\aalg}$
  together with $\undefd{\interpreta{\ell}{\alpha}}$ for all $\ell \to r \in \atrs'$
  and $\alpha$ it follows that the rules in $\atrs'$ are not reachable. All terms in $\family{\termset}$ are normal
  forms with respect to $\atrs'$. Hence we can ignore these rules.
\end{proof}

\begin{example}{
  \newcommand{\ssuc}{\msf{s}}
  \newcommand{\suc}{\funap{\msf{s}}}
  \newcommand{\zer}{0}
  Consider the TRS $\atrs$ consisting of the following four rules:
  \begin{align*}
    \funap{f}{\suc{\suc{x}}} &\to \funap{f}{\funap{o}{x}} &
    \funap{o}{\suc{\suc{x}}} &\to \suc{\suc{\funap{o}{x}}} &
    \funap{o}{\zer} &\to \zer &
    \funap{o}{\suc{\zer}} &\to \suc{\suc{\suc{\zer}}}
  \end{align*}
  The TRS is not terminating: $\funap{f}{\suc{\suc{\suc{\zer}}}} \red \funap{f}{\funap{o}{\suc{\zer}}} \red \funap{f}{\suc{\suc{\suc{\zer}}}} \to \ldots$. However, the function $f$ is terminating when applied to an even number,
  that is, the language $\termset = \{\funap{f}{\funap{\ssuc^{2\cdot n}}{\zer}} \where n \in \nat\}$.
  We choose $\aalg = \triple{\{0,1\}}{\sinterpret}{{\qge}}$
  where $\interpret{\zer} = 0$, $\funap{\interpret{\ssuc}}{0} = 1$, $\funap{\interpret{\ssuc}}{1} = 0$,
  $\funap{\interpret{o}}{0} = 0$, $\undefd{\funap{\interpret{o}}{1}}$, $\funap{\interpret{f}}{0} = 0$ and $\undefd{\funap{\interpret{f}}{1}}$.
  Then $\aalg$ is a monotone partial model with $\termset \subseteq \lang{\aalg}$.
  We have $\undefd{\interpreta{\funap{o}{\suc{\zer}}}{\alpha}}$ (for all $\alpha$), thus the rule
  $\funap{o}{\suc{\zer}} \to \suc{\suc{\suc{\zer}}}$ is never applicable and can be removed.
}\end{example}

%% file: conclusion.tex
\noindent
We have implemented some of the methods proposed in this paper.
More information and the source code of the implementations can be found on the website:
\begin{center}
  \url{http://infinity.few.vu.nl/local/}
\end{center}
In particular, we have implemented the method from Section~\ref{sec:s}.
The program automatically finds the minimal partial model $\aalg$ for the language of normalizing $\combS$-terms,
and transforms the local termination problem into a global termination problem.
We have formally verified the model property, and that
all terms that are not in the language of $\aalg$  are non-terminating.
Global termination of the transformed system 
has been proven by TTT2~(1.0) \cite{TTT09} and formally verified by CeTA~(1.05) \cite{ceta}.
Thereby we have automated one of the central contributions of~\cite{DBLP:journals/iandc/Waldmann00}.

We intend to generalize the characterization of local termination to context-sensitive rewriting~\cite{luca:98}, using $\samumap$-monotonic, partial $\asig$-algebras; and also 
to top termination, using weakly extended, monotone, partial $\asig$-algebras~\cite{AG00,EWZ06}.

Methods using transformations from certain properties, 
like liveness properties~\cite{adam} or outermost termination~\cite{zant:outermost:08},
to termination usually give rise to local termination problems.
That is, termination is of interest only for those terms 
which are in the image of the transformation.
For example, we noted that the transformation in~\cite{zant:outermost:08} gives rise 
to a language which can be described by a partial model.
Then it suffices to show completeness of the transformation to local termination,
and employing Theorem~\ref{thm:transform} we obtain a complete transformation 
to global termination for free.
\weg{%
Noteworthy is also that using Theorem~\ref{thm:transform}
we obtain a variant of the transformation which has not been considered 
in~\cite{zant:outermost:08},
yielding the same number of rules with a larger signature.
The symbols $f^{\sharp}$ are additionally labeled with a number $i$ 
(that is, $f^{\sharp,i}$)
which indicates the number of the argument to which $\msf{down}$ has descended to.
This larger signature gives more freedom for the choice of interpretations
which could in principle strengthen the transformation.
Experiments remain to be carried out.
}%

\paragraph{Acknowledgements}
We than Vincent van Oostrom and the anonymous referees
for valuable suggestions for improving the presentation of the paper.